\documentclass[letterpaper, 10 pt, conference]{ieeeconf}
\IEEEoverridecommandlockouts
\onecolumn
\usepackage[english]{babel}
\usepackage{theorem}
\theoremheaderfont{\itshape\bfseries}
{\theorembodyfont{\itshape}
	\newtheorem{assumption}{\textbf{Assumption}}
	\newtheorem{lemma}{\textbf{Lemma}}
	\newtheorem{definition}{\textbf{Definition}}
	\newtheorem{theorem}{\textbf{Theorem}}
	
	\newtheorem{remark}{\textbf{Remark}}
	
	\newtheorem{problem}{\textbf{Problem}}
}
\usepackage{graphicx}
\usepackage{newtxtext,newtxmath}
\usepackage{amsmath}
\usepackage{amsfonts}
\usepackage{mathrsfs}
\usepackage{xcolor}
\usepackage{graphicx}
\usepackage{tcolorbox}
\usepackage{mdframed}
\usepackage{booktabs}
\usepackage{booktabs}
\usepackage{caption}
\makeatletter
\renewcommand\section{\@startsection {section}{1}{\z@}%
	{-3.5ex \@plus -1ex \@minus -.2ex}%
	{2.3ex \@plus.2ex}%
	{\normalfont\scshape\centering\bfseries}}
\renewcommand\subsection{\@startsection {subsection}{1}{\z@}%
	{-3.5ex \@plus -1ex \@minus -.2ex}%
	{2.3ex \@plus.2ex}%
	{\normalfont\itshape\bfseries}}

\newcommand{\T}{^{\mbox{\tiny T}}}
\newcommand{\R}{\mathbb{R}}

\let\leq\leqslant
\let\geq\geqslant

\newenvironment{proof}[1][Proof]%
{\par\noindent\textit{#1:\ }}%
{\hspace*{\fill} \rule{6pt}{6pt}}
\newenvironment{proof*}[1][Proof]%
{\par\noindent\textit{#1:\ }}{}

\DeclareMathOperator{\diag}{diag}

\usepackage{tikz}
\usetikzlibrary{shapes,calc,arrows,patterns,decorations.pathmorphing
	,decorations.markings}
\usetikzlibrary{arrows.meta}

\newenvironment{system}[1]%
{\setlength{\arraycolsep}{0.5mm}\left\{ \; \begin{array}{#1}}%
	{\end{array} \right.}
\newenvironment{system*}[1]%
{\setlength{\arraycolsep}{0.5mm} \begin{array}{#1}}%
	{\end{array}}

\title{\LARGE \textbf{Scale-free Design for Delayed Regulated Synchronization of Homogeneous and Heterogeneous Discrete-time Multi-agent Systems Subject to Unknown Non-uniform and Arbitrarily Large Communication Delays}}

\author{Donya Nojavanzadeh, Zhenwei Liu, Ali Saberi and Anton A. Stoorvogel%
 \thanks{Donya Nojavanzadeh is with
	School of Electrical Engineering and Computer Science, Washington
	State University, Pullman, WA 99164, USA {\tt\small 
		donya.nojavanzadeh@wsu.edu}} 
	\thanks{Zhenwei Liu is with College of Information Science and
		Engineering, Northeastern University, Shenyang 110819,
		P. R. China {\tt\small jzlzwsy@gmail.com}} \thanks{Ali Saberi is with
		School of Electrical Engineering and Computer Science, Washington
		State University, Pullman, WA 99164, USA {\tt\small
			saberi@eecs.wsu.edu}} \thanks{Anton A. Stoorvogel is with
		Department of Electrical Engineering, Mathematics and Computer
		Science, University of Twente, P.O. Box 217, Enschede, The
		Netherlands {\tt\small A.A.Stoorvogel@utwente.nl}}}

\begin{document}
	
	\maketitle
	\thispagestyle{empty}
	\pagestyle{empty}
	
	\begin{abstract}
In this paper, we study delayed regulated state/output synchronization for discrete-time homogeneous and heterogeneous networks of multi-agent systems (MAS) subject to unknown, non-uniform and arbitrarily large communication delays. A delay transformation is utilized to transform the original MAS to a new system without delayed states. The proposed \emph{scale-free} dynamic protocols are developed solely based on agent models and localized information exchange with neighbors such that we do not need any information about the communication networks and the number of agents.
	\end{abstract}

	\section{Introduction}

Cooperative control of multi-agent systems (MAS) such as synchronization, consensus, swarming, flocking, has become a hot topic among researchers because of its broad application in various areas such as biological systems, sensor networks, automotive vehicle control, robotic cooperation teams, and so on. See for example books \cite{ren-book,wu-book,kocarev-book} or the survey paper \cite{saber-murray3}.

State synchronization inherently requires homogeneous networks. Most works have focused on state synchronization where each agent has access to a linear combination of its own state relative to that of the neighboring agents, which is called full-state coupling; see \cite{saber-murray3,saber-murray,saber-murray2,ren-atkins,ren-beard-atkins,tuna1}. A more realistic scenario which is partial-state coupling (i.e. agents share part of their information over the network) is studied in \cite{tuna2,li-duan-chen-huang,pogromsky-santoboni-nijmeijer,tuna3}. On the other hand, for heterogeneous network it is more reasonable to consider output synchronization since the dimensions of states and their physical interpretation may be different. For heterogeneous MAS with non-introspective agents \footnote{Agents are said to be introspective when they have access to either exact or estimates of their states, otherwise they are called non-introspective \cite{grip-yang-saberi-stoorvogel-automatica}.}, it is well-known that one needs to regulate outputs of the agents to \emph{a priori} given trajectory generated by a so-called exosystem (see
\cite{wieland-sepulchre-allgower, grip-saberi-stoorvogel3}). Other works on synchronization of MAS with non-introspective agents can be found in the literature as  \cite{grip-yang-saberi-stoorvogel-automatica,grip-saberi-stoorvogel}. Most of the literature for heterogeneous MAS with introspective agents are based on modifying the agent dynamics via local feedback to achieve some form of homogeneity. There have been many results for synchronization of heterogeneous networks with introspective agents, see for instance \cite{kim-shim-seo,yang-saberi-stoorvogel-grip-journal,li-soh-xie-lewis-TAC2019,modares-lewis-kang-davoudi-TAC2018,qian-liu-feng-TAC2019,chen-auto2019}.

In real applications, networks may be
subject to delays. Time delays may afflict system performance or even
lead to instability. As discussed in
\cite{cao-yu-ren-chen}, two kinds of delays have been considered in the
literature: input delays and communication delays. Input delays encapsulate the
processing time to execute an input for each agent, whereas
communication delays can be considered as the time it takes to transmit
information from an agent to its destination. Many works have been focused on dealing with input
delays, specifically with the objective of deriving an upper bound on the input delays
such that agents can still achieve synchronization. See, for example
\cite{ferrari-trecate,lin-jia,tian-liu,saber-murray2,xiao-wang}. Some research has
been done for networks subject to communication delays. Fundamentally, there are two approaches in the literature for dealing with MAS subject to communication delays.
\begin{enumerate}
	\item Standard output synchronization subject to regulating output to a constant trajectory.
	\item Delayed state/output synchronization.
\end{enumerate}

Both of these approaches preserves diffusiveness of the couplings (i.e. ensuring the invariance of the consensus manifold). Also, the notion of the delayed output synchronization coincides with the standard output synchronization if it is required that output regulated to a constant trajectory. As such delayed synchronization can be viewed as the generalization of standard synchronization in the context of MAS subject to communication delay.

Majority of research on MAS subject to communication delay have been focused on achieving the standard output synchronization by regulating the output to constant trajectory \cite{ghabcheloo,klotz-obuz-kan,munz-papachristodoulou-allgower,munz-papachristodoulou-allgower2,tian-liu,xiao-wang-tac,zhang-saberi-stoorvogel-ejc,liang-com-delay,yao-com-delay,wang-com-delay}. It is worth noting that in all of the aforementioned papers, design of protocols require knowledge of the graph and size of the network. More recently, the notion of delayed synchronization was introduced in \cite{chopra-spong2} for MAS with passive agents and strongly connected and balanced graphs where it is assumed that there exists a unique path between any two distinct nodes. Then, the authors extended their results in \cite{chopra-tac,chopra-spong-cdc06} when they allowed multiple paths between two agents in strongly connected communication graphs. Although the synchronized trajectory in these papers is constant and standard definition of synchronization can be utilized, the authors motivation for utilizing delayed synchronization is exploring the possible existence of delayed-induced periodicity in synchronized trajectory of coupled systems. These solutions, provided they exist, can be valuable in several applications, for example \cite{pyragas,vicente}. It is worth to note that the protocol design in these papers do not need knowledge of the graph, since they are restricted to passive agents. An interesting line of research utilizing delayed synchronization formulation was introduced recently in \cite{Liu-Saberi-Stoorvogel-Li_delayed-con,Liu-Saberi-Stoorvogel-Li_delayed-dis}. These papers considered a \textbf{dynamic} synchronized trajectory (i.e. non constant synchronized trajectory). They designed protocols to achieve regulated delayed state/output synchronization in presence of communication delays where the communication graph was spanning tree. However, the protocol design required knowledge of the graph and size of the network.

In this paper, we extend our earlier results of delayed synchronization by developing \textbf{scale-free} framework utilizing localized information exchange for homogeneous and heterogeneous MAS subject to unknown non-uniform and arbitrarily large communication delays to achieve delayed regulated synchronization when the synchronized trajectory is a \emph{dynamic} signal generated by a so-called exosystem. The associated graphs to the communication networks are assumed to be directed spanning tree (i.e., they have one root node and the other non-root nodes have in-degree one). We achieve scale-free delayed regulated state synchronization for discrete-time homogeneous MAS with non-introspective agents, and scale-free delayed regulated output synchronization for discrete-time heterogeneous MAS with introspective agents. Our proposed design methodologies are scale-free, namely

\begin{itemize}
	\item The design is independent of information about the communication network such as the spectrum of the associated Laplacian matrix or size of the network.
	\item The collaborative protocols will work for any network with associated directed spanning tree, and can tolerate any unknown, non-uniform and arbitrarily large communication delays.
\end{itemize}


\subsection*{Notations and definitions}

Given a matrix $A\in \mathbb{R}^{n\times m}$, $A\T$ denotes its
conjugate transpose and $\|A\|$ is the induced 2-norm. Let $j$ indicate $\sqrt{-1}$. A square matrix
$A$ is said to be Schur stable if all its eigenvalues are in the closed unit disc. We denote by
$\diag\{A_1,\ldots, A_N \}$, a block-diagonal matrix with
$A_1,\ldots,A_N$ as the diagonal elements.  $I_n$ denotes the
$n$-dimensional identity matrix and $0_n$ denotes $n\times n$ zero
matrix; sometimes we drop the subscript if the dimension is clear from
the context. For $A\in \mathbb{C}^{n\times m}$ and $B\in \mathbb{C}^{p\times q}$, the Kronecker product of $A$ and $B$ is defined as
\[
A\otimes B=\begin{pmatrix}
a_{11}B& \hdots& a_{1m}B\\
\vdots&\vdots&\vdots\\
a_{n1}B&\hdots&a_{nm}B
\end{pmatrix}.
\]
The following property of the Kronecker product will be particularly useful
\[
(A\otimes B)(C\otimes D)=(AC)\otimes (BD).
\]

To describe the information flow among the agents we associate a \emph{weighted graph} $\mathcal{G}$ to the communication network. The weighted graph $\mathcal{G}$ is defined by a triple
$(\mathcal{V}, \mathcal{E}, \mathcal{A})$ where
$\mathcal{V}=\{1,\ldots, N\}$ is a node set, $\mathcal{E}$ is a set of
pairs of nodes indicating connections among nodes, and
$\mathcal{A}=[a_{ij}]\in \mathbb{R}^{N\times N}$ is the weighted adjacency matrix with non negative elements $a_{ij}$. Each pair in $\mathcal{E}$ is called an \emph{edge}, where
$a_{ij}>0$ denotes an edge $(j,i)\in \mathcal{E}$ from node $j$ to
node $i$ with weight $a_{ij}$. Moreover, $a_{ij}=0$ if there is no
edge from node $j$ to node $i$. We assume there are no self-loops,
i.e.\ we have $a_{ii}=0$. A \emph{path} from node $i_1$ to $i_k$ is a
sequence of nodes $\{i_1,\ldots, i_k\}$ such that
$(i_j, i_{j+1})\in \mathcal{E}$ for $j=1,\ldots, k-1$. A \emph{directed tree} is a subgraph (subset
of nodes and edges) in which every node has exactly one parent node except for one node, called the \emph{root}, which has no parent node. The \emph{root set} is the set of root nodes. A \emph{directed spanning tree} is a subgraph which is
a directed tree containing all the nodes of the original graph. If a directed spanning tree exists, the root has a directed path to every other node in the tree. For a weighted graph $\mathcal{G}$, the matrix
$L=[\ell_{ij}]$ with
\[
\ell_{ij}=
\begin{system}{cl}
\sum_{k=1}^{N} a_{ik}, & i=j,\\
-a_{ij}, & i\neq j,
\end{system}
\]
is called the \emph{Laplacian matrix} associated with the graph
$\mathcal{G}$. The Laplacian matrix $L$ has all its eigenvalues in the
closed right half plane and at least one eigenvalue at zero associated
with right eigenvector $\textbf{1}$ \cite{royle-godsil}. Moreover, if the graph contains a directed spanning tree, the Laplacian matrix $L$ has a single eigenvalue at the origin and all other eigenvalues are located in the open right-half complex plane \cite{ren-book}.
A matrix $D=\{d_{ij}\}_{N\times N}$ is called a row stochastic matrix if 
\begin{enumerate}
	\item $d_{ij}\ge0$ for any $i,j$,
	\item $\sum_{j=1}^{N}d_{ij}=1$ for $i=1,\hdots, N$.
\end{enumerate}
A row stochastic matrix $D$ has at least one eigenvalue at $1$ with right eigenvector $1$. $D$ can be associated with a graph $\mathcal{G}=(\mathcal{V}, \mathcal{E}, \mathcal{A})$. The number of nodes $N$ is the dimension of $D$ and an edge $(j,i)\in \mathcal{E}$ if $d_{ij}>0$.

\section{Homogeneous MAS with non-introspective agents}
 Consider a MAS consists of $N$ identical linear agents
\begin{equation}\label{agent}
  \begin{system}{ccl}
    {x}_i(k+1) &=& Ax_i(k) +B u_i(k)\\
    y_i (k)&=& Cx_i(k)
  \end{system}
\end{equation}
for $i\in \{1,\ldots,N\}$, where $x_i \in \R^{n}$, $y_i \in \R^{p}$, and $u_i \in \R^m$ are the state, output and the input of
agent $i$, respectively.

We make the following assumption on agent models.
\begin{assumption}\label{agentass} 
	All eigenvalues of $A$ are in the closed unit disc. Moreover, $(A,B)$ is stabilizable and $(A,C)$ is detectable.
\end{assumption}

The network provides agent $i$ with the following information
\begin{equation}\label{zeta}
\zeta_{i}(k)=\dfrac{1}{1+\sum_{j=1}^{N}a_{ij}}\sum_{j=1}^N a_{ij}(y_{i}(k)-y_{j}(k-\kappa_{ij})),
\end{equation}
where $\kappa_{ij} \in\mathbb{N}^{+}$ represents an unknown communication delay from agent $j$ to agent $i$. In the above
$a_{ij}\geq 0$. This communication topology of the network, presented in \eqref{zeta}, can be associated to a weighted graph $\mathcal{G}$ with each
node indicating an agent in the network and the weight of an edge is
given by the coefficient $a_{ij}$.  The communication
delay implies that it took $\kappa_{ij}$ seconds for
agent $j$ to transfer its state information to agent $i$.

Next we write $\zeta_i$ as
\begin{equation}\label{zeta-y}
\zeta_i(k)=\sum_{j=1}^N d_{ij}(y_i(k)-y_j(k-\kappa_{ij})),
\end{equation}
where $d_{ij}\geq 0$, and we choose
$d_{ii}=1-\sum_{j=1,j\neq i}^Nd_{ij}$ such that $\sum_{j=1}^Nd_{ij}=1$
with $i,j\in\{1,\ldots,N\}$. Note that $d_{ii}$ satisfies $d_{ii}>0$. The weight matrix $D=[d_{ij}]$ is then a so-called, row stochastic matrix. Let $D_{in}=\diag\{d_{in}(i)\}$ with
$d_{in}(i)=\sum_{j=1}^{N}a_{ij}$. Then the relationship between the
row stochastic matrix $D$ and the Laplacian matrix $L$ is
\begin{equation}\label{hodt-LD}
(I+D_{in})^{-1}L=I-D.
\end{equation}

We refer to \eqref{zeta} as \emph{partial-state coupling} since only part of
the states are communicated over the network. When $C=I$, it means all states are communicated over the network and we call it \emph{full-state coupling}. \\

We make the following definition. 
\begin{definition}\label{ungrN}
	Let $\mathbb{G}^N$ denote the set of
	directed spanning tree graphs with $N$ nodes for which the corresponding Laplacian matrix $L$ is
	lower triangular. The corresponding Laplacian matrix $L$ has the property that the entries of the first row are equal to zero and $\ell_{ii}>0$ for $i=2,\ldots,N$. We consider agent $1$ as the root agent. 
\end{definition}

\begin{remark}
	Note that any graph which is a directed spanning tree, has a possible lower triangular Laplacian matrix after reordering of the agents.
\end{remark}

	For the graph defined by Definition \ref{ungrN}, we have that row stochastic matrix $D$ is lower triangular matrix with $d_{11}=1$ and $d_{1j}=0$ for $j=2,\hdots, N$. Therefore, we have
\[
D=\begin{pmatrix}
1 & 0 & 0 & \cdots & 0 \\
d_{21} & d_{22} & 0 & \cdots & 0 \\
d_{31} & d_{32} & d_{33} & \ddots & \vdots \\
\vdots & \ddots & \ddots & \ddots & 0\\
d_{N1} & \cdots & d_{N,N-2} & d_{N,N-1} & d_{N,N}
\end{pmatrix}.
\]
Since the graph is equal to a directed spanning tree, in
every row (except the first one) there are exactly two elements unequal to $0$. 

Our goal is to achieve delayed regulated state synchronization among all agents while the synchronized dynamics are equal to a priori given trajectory generated by a so-called exosystem
\begin{equation}\label{exo}
\begin{system*}{cl}
{x}_r(k+1)&=Ax_r(k), \quad x_r(0)=x_{r0}\\
y_r(k)&=Cx_r(k)
\end{system*}
\end{equation}
where $x_r\in \mathbb{R}^n$ and $y_r\in\mathbb{R}^p$. 

Clearly, we need some level of communication between the constant trajectory and the agents. According to the structure of communication network, we assume that each agent has access to the quantity
\begin{equation}
\psi_i=\iota_i(y_i(k)- y_r(k-\kappa_{ir})), \qquad \iota_i=\begin{system}{cl}
1, &\quad i=1,\\
0, &\quad i=2,\cdots,N.
\end{system}
\end{equation}
	By combining this with \eqref{zeta-y}, the information exchange among agents is given by

\begin{equation}\label{zetabar}
\bar{\zeta}_i(k)=\sum_{j=1}^{N}a_{ij}(y_i(k)-y_j(k-\kappa_{ij}))+\iota_i(y_i(k)-y_r(k-\kappa_{ir})).
\end{equation}

For any graph $\mathbb{G}^N$, with the Laplacian matrix $L$, we define the expanded Laplacian matrix as: 
\[
\bar{L}=L+diag\{\iota_i\}=[\bar{\ell}_{ij}]_{N \times N}
\]
which is not a regular Laplacian matrix associated to the graph, since the sum of its rows need not be zero.
In terms of the coefficients of the expanded Laplacian matrix $\bar{L}$, $\bar{\zeta}_i$ in \eqref{zetabar} can be rewritten as:
\begin{equation}\label{zetabar2}
\bar{\zeta}_i(k)=\frac{1}{2+d_{in}(i)}\sum_{j=1}^{N}\bar{\ell}_{ij}(y_j(k-\kappa_{ij})-y_r(k-\kappa_{ir}))=y_i(k)-y_r(k-\kappa_{ir})-\sum_{j=1}^{N}\bar{d}_{ij}(y_j(k-\kappa_{ij})-y_r(k-\kappa_{ir}))
\end{equation}
and we define
\begin{equation}\label{bar-D}
\bar{D}=I-(2I+D_{in})^{-1}\bar{L}.
\end{equation}

It is easily verified that the matrix $\bar{D}$ is a matrix with all elements non negative and the sum of each row is less than or equal to $1$. 

In this paper, we also introduce a localized information exchange among protocols. In particular, each agent $i=1,\hdots N$ has access to localized information denoted by $\hat{\zeta}_i$, of the form 

\begin{equation}\label{info2}
\hat{\zeta}_i(k)=\frac{1}{2+d_{in}(i)}\sum_{j=1}^N\bar{\ell}_{ij}\xi_j(k-\hat{\kappa}_{ij})=\xi_i(k)-\sum_{j=1}^{N}\bar{d}_{ij}\xi_j(k-\hat{\kappa}_{ij})
\end{equation}

where $\xi_j \in \mathbb{R}^n$ is a variable produced internally by agent $j$ and to be defined in next sections while $\hat{\kappa}_{ij}\in\mathbb{N}^{+}$ ($i\neq j$) represents an unknown communication delay from agent $j$ to agent $i$.

We define the following definition.

\begin{definition}\label{def-synch}
	The agents of a MAS are said to achieve 
	\begin{itemize}
		\item delayed state  synchronization if 
		\begin{equation}\label{delayed-x-sync}
		\lim_{k\to\infty}\left[(x_i(k)-x_j(k-{\kappa}_{ij})\right]=0, \quad \text{for all } i,j\in \{1,\hdots,N\}.
		\end{equation}
		where ${\kappa}_{ij}$ represents communication delay from agent $j$ to agent $i$.
		
		\item and delayed regulated state synchronization if 
		\begin{equation}\label{delayed-reg-sync}
		\lim_{k\to\infty}\left[(x_i(k)-x_r(k-{\kappa}_{ir})\right]=0, \quad \text{for all } i\in \{1,\hdots,N\}.
		\end{equation}
		where $\kappa_{ir}$ represents the sum of delays from agent $i$ to the exosystem.
	\end{itemize}
\end{definition}

We formulate the following problem of delayed state synchronization for networks with unknown, nonuniform communication delays with linear dynamic protocols as follows.

\begin{problem}\label{prob-x}
Consider a MAS described by \eqref{agent} and  \eqref{zetabar2}. Let
$\mathbb{G}^N$ be the set of network graphs as defined in Definition \ref{ungrN}.

Then, the \textbf{scalable delayed regulated state synchronization problem based on localized information exchange utilizing collaborative protocols} for networks with unknown, non-uniform and arbitrarily large communication delay is to find, if possible, a linear dynamic protocol for each agent $i \in \{1,\hdots,N\}$, using only knowledge of agent model, i,e. $(A,B,C)$, of the form:
\begin{equation}\label{pro}
\begin{system}{cl}
{x}_{c,i}(k+1)&=A_c x_{c,i}(k)+B_{c1} \bar{\zeta}_i(k)+B_{c2} \hat{\zeta}_i(k),\\
u_i(k)&=F_c x_{c,i}(k),
\end{system}
\end{equation}
where $\hat{\zeta}_i(k)$ is defined in \eqref{info2} with $\xi_i(k)=H_c x_{c,i}(k)$ and $x_{c,i}\in \mathbb{R}^{n_c}$
such that for any $N$, any graph $\mathscr{G}\in \mathbb{G}^N$, any communication delays $\kappa_{ij}$ and $\hat{\kappa}_{ij}$ we achieve delayed regulated state synchronization as stated by \eqref{delayed-reg-sync} in Definition \ref{def-synch}.

\end{problem}

\subsection{Protocol Design}
In this section, we provide our results for scalable delayed regulated state synchronization of MAS with full- and partial-state coupling.

\subsubsection{\textbf{Full-state coupling}}
First we consider MAS with full-state coupling, i.e. $C=I$.\\

\begin{table}[h]
	\centering
    \captionsetup[table]{labelformat=empty}
    \caption*{Protocol 1. Full-state Coupling}
	\begin{tabular}{p{15cm}}
		\toprule
       We design collaborative protocols based on localized information exchanges for agents $i=1,\hdots, N$ as
	\begin{equation}\label{fscps}
\begin{system}{cl}
{\chi}_i(k+1)&=A\chi_i(k)+Bu_i(k)+A\bar{\zeta}_i(k)-A\hat{\zeta}_i(k),\\
u_i(k)&=-K\chi_i(k),
\end{system}
\end{equation}
	where $\bar{\zeta}_i(k)$ is defined by \eqref{zetabar2} and $\hat{\zeta}_i$ is given by
\begin{equation}\label{add_1}
\hat{\zeta}_i(k)=\chi_i(k)-\sum_{j=1}^{N}\bar{d}_{ij}\chi_j(k-\hat{\kappa}_{ij})
\end{equation}
which means the agents communicate $\xi_i(k)=\chi_i(k)$. Matrix $K$ is designed such that $A-BK$ is Schur stable.\\
		\bottomrule
	\end{tabular}
\end{table} 

Then, we have the following theorem for scalable delayed regulated state synchronization of MAS with full-state coupling.
\begin{theorem}\label{thm_f}
	Consider MAS \eqref{agent}, with $C=I$, consisting of $N$ agents satisfying Assumption \ref{agentass}. Let the associated network communication be given by \eqref{zetabar2}.  
	
	Then, the scalable delayed regulated state synchronization problem as defined in Problem \ref{prob-x} is solvable. In particular, the linear dynamic protocol \eqref{fscps} solves delayed regulated state synchronization problem for any $N$ and any graph
		$\mathcal{G}\in\mathbb{G}^N$.
\end{theorem} 

\begin{proof}[Proof of Theorem \ref{thm_f}] First, we define
\[
\bar{x}_i(k)=x_i(k+{\kappa}_{ir})\text{ and } \bar{\chi}_i(k)=\chi_i(k+{\kappa}_{ir})
\]
where ${\kappa}_{ir}$ denotes the sum of delays from agent $i$ to the exosystem, and $\kappa_{ij}={\kappa}_{ir}-{\kappa}_{jr}$. Then we have
\begin{equation}\label{zeta-bar-bar}
\begin{system*}{cl}
\bar{\bar{\zeta}}_{i}(k)=&\bar{\zeta}_{i}(k+{\kappa}_{ir})=x_i(k+{\kappa}_{ir})-x_r(k)\\
&-\sum_{j=1}^{N}\bar{d}_{ij}(x_j(k+{\kappa}_{ir}-\kappa_{ij})-x_r(k))\\
&=\bar{x}_i(k)-x_r(k)-\sum_{j=1}^{N}\bar{d}_{ij}(\bar{x}_j(k)-{x}_r(k))
\end{system*}
\end{equation}
and 
\begin{equation}\label{zeta-hat-bar}
\begin{system*}{cl}
\hat{\bar{\zeta}}_{i}=\hat{\zeta}_{i}(k+{\kappa}_{ir})=\chi_i(k+\kappa_{ir})-\sum_{j=1}^{N}\bar{d}_{ij}\chi_j(k+\kappa_{ir}-\hat{\kappa}_{ij})\\
=\bar{\chi}_i(k)-\sum_{j=1}^{N}\bar{d}_{ij}\bar{\chi}_j(k+{\kappa}_{ij}-\hat{\kappa}_{ij}).
\end{system*}
\end{equation}
Then, by defining $\tilde{x}_i(k)=\bar{x}_i(k)-{x}_r(k)$ and
\[
\tilde{x}=\begin{pmatrix}
\tilde{x}_1\\\vdots \\ \tilde{x}_N
\end{pmatrix}, \quad \bar{\chi}=\begin{pmatrix}
\bar{\chi}_1\\\vdots \\ \bar{\chi}_N
\end{pmatrix}
\]
 we have the following closed-loop system in frequency domain as
\[
\begin{system}{cl}
e^{j\omega}{\tilde{x}}=&(I\otimes A) \tilde{x}-(I\otimes BK)\bar{\chi}\\
e^{j\omega}{\bar{\chi}}=&\left(I\otimes(A-BK)\right)\bar{\chi}+((I-\bar{D})\otimes A)\tilde{x}\\
&-( (I-\bar{D}_{j\omega}(\kappa))\otimes A)\bar{\chi}
\end{system}
\]
where
\begin{equation*}
\begin{system*}{cl}
\bar{D}_{j\omega}(\kappa)=
\begin{pmatrix}
\bar{d}_{11}&0&0&\cdots&0\\
\bar{d}_{21}e^{-j\omega (\hat{\kappa}_{21}-\kappa_{21})}&\bar{d}_{22}&0&\cdots&0\\
\vdots&\cdots&\ddots&\ddots&\vdots\\
\bar{d}_{N1}e^{-j\omega (\hat{\kappa}_{N1}-\kappa_{N1})}&\bar{d}_{N2}e^{-j\omega (\hat{\kappa}_{N2}-\kappa_{N2})}&\cdots&\cdots&\bar{d}_{NN}
\end{pmatrix}.
\end{system*}
\end{equation*}
Let $\delta=\tilde{x}-\bar{\chi}$. Then, we can obtain,
\begin{equation}
\begin{system}{cl}\label{error-full}
e^{j\omega}{\tilde{x}}&=(I\otimes (A-BK)) \tilde{x}+(I\otimes BK)\delta\\
e^{j\omega} \delta&=(\bar{D}_{j\omega}(\kappa)\otimes A)\delta+\left((\bar{D}-\bar{D}_{j\omega}(\kappa))\otimes A\right)\tilde{x}
\end{system}
\end{equation}

We prove \eqref{error-full} is asymptotically stable for all communication
delays ${\kappa}_{ij}\in \mathbb{N}^+$ and $\hat{\kappa}_{ij}\in \mathbb{N}^+$. Following the critical lemma \cite[Lemma 3]{zhang-saberi-stoorvogel-delay}, we first prove stability without
communication delays ${\kappa}_{ij}$ and $\hat{\kappa}_{ij}$ and then prove stability in presence of communication delays.
\begin{itemize}
	\item In the absence of communication delays in the network,
the stability of system \eqref{error-full} is equivalent to the stability of matrix
\begin{equation}
\begin{pmatrix}
	I\otimes (A-BK)&I\otimes BK\\
	0&\bar{D}\otimes A
\end{pmatrix}
\end{equation}
	where $\bar{D}=[\bar{d}_{ij}]\in\R^{N\times N}$ and we have that the eigenvalues of $\bar{D}$  are in open unit disk.
The eigenvalues of $\bar{D}\otimes A$ are of the form
$\lambda_i \mu_j$, with $\lambda_i$ and $\mu_j$ eigenvalues of
$\bar{D}$ and $A$, respectively \cite[Theorem 4.2.12]{horn-johnson}. Since $|\lambda_i|<1$ and
$|\mu_j|\leq 1$, we find $\bar{D}\otimes A$ is Schur stable. Then we have
\begin{equation}\label{estable}
\lim_{k\to \infty}\delta_i(k)\to 0
\end{equation}
Therefore, we have that the dynamics for $\delta_i(k)$ is asymptotically stable.

 \item  In the presence of communication delay, the closed-loop system \eqref{error-full} is asymptotically stable if
\begin{equation}\label{cond-full}
\det\left[e^{j\omega} I-\begin{pmatrix}
I\otimes (A-BK)&I\otimes BK\\
(\bar{D}-\bar{D}_{j\omega}(\kappa))\otimes A&\bar{D}_{j\omega}(\kappa)\otimes A
\end{pmatrix}\right]\ne 0
\end{equation}
for all $\omega \in \mathbb{R}$ and any communication delays ${\kappa}_{ij}\in \mathbb{N}^+$ and $\hat{\kappa}_{ij}\in \mathbb{N}^+$.
 Condition \eqref{cond-full} is satisfied if matrix
 \begin{equation}\label{cond-matrix}
 \begin{pmatrix}
 I\otimes (A-BK)&I\otimes BK\\
(\bar{D}-\bar{D}_{j\omega}(\kappa))\otimes A&\bar{D}_{j\omega}(\kappa)\otimes A
 \end{pmatrix}
 \end{equation}
 has no eigenvalues on the unit circle for all $\omega \in \mathbb{R}$. That is to say we just need to prove the stability of 
 \begin{equation}
\begin{system}{cl}\label{error-full-2}
e^{j\omega}{\tilde{x}}&=(I\otimes (A-BK)) \tilde{x}+(I\otimes BK)\delta\\
e^{j\omega} \delta&=(\bar{D}_{j\omega}(\kappa)\otimes A)\delta+\left((\bar{D}-\bar{D}_{j\omega}(\kappa))\otimes A\right)\tilde{x}
\end{system}
 \end{equation}
 Then, according to the structure of matrix $\bar{D}$, \eqref{error-full-2} can be rewritten as
  \begin{equation}
 \begin{system}{cl}\label{x-1}
 e^{j\omega} \tilde{x}_1&=  (A-BK)\tilde{x}_1+ BK\delta_1\\
 e^{j\omega} \delta_1&=\bar{d}_{11}A\delta_1
 \end{system}
 \end{equation}
 and
   \begin{equation}\label{x-i}
 \begin{system}{cl}
 e^{j\omega} \tilde{x}_i&=  (A-BK)\tilde{x}_i+ BK\delta_i\\
 e^{j\omega} \delta_i&=\bar{d}_{ii}A\delta_i+\sum_{j=1}^{i-1}\bar{d}_{ij}e^{-j\omega({\hat{\kappa}_{ij}-\kappa}_{ij})}A\delta_j\\
 \qquad&+\sum_{j=1}^{i-1}(1-e^{-j\omega{(\hat{\kappa}_{ij}-{\kappa}_{ij})}})\bar{d}_{ij}A\tilde{x}_j
 \end{system}
 \end{equation}
 for $i=2,\hdots, N$ and $j\le i-1$. 
 
 Then for $i=1$, The eigenvalues of $\bar{d}_{11}A$ are of the form $\bar{d}_{11}\lambda_i$, with $\lambda_i$ eigenvalues of $A$. since $|\bar{d}_{11}|<1$, and $|\lambda_i|\leq1$, one can obtain that all eigenvalues of $\bar{d}_{11}A$ are inside unit circle, that is 
 \[
 \delta_1\to 0 \text{ as } k\to \infty
 \]
 then, given that $A-BK$ is Schur stable, we have
  \[
 \tilde{x}_1\to 0 \text{ as } k\to \infty
 \]
 Therefore, the dynamics of $\tilde{x}_1$, and $\delta_1$ are asymptotically stable.
 
 Then for $i=2$ and $j=1$, we have 
\begin{equation}
  \begin{system}{cl}
 	e^{j\omega} \tilde{x}_2&=  (A-BK)\tilde{x}_2+ BK\delta_2\\
 	e^{j\omega} \delta_2&=\bar{d}_{22}A\delta_2+\bar{d}_{21}e^{-j\omega(\hat{\kappa}_{21}-{\kappa}_{21})}A\delta_1\\
 	&\hspace{2cm}+(1-e^{-j\omega{(\hat{\kappa}_{21}-{\kappa}_{21})}})\bar{d}_{21}A\tilde{x}_1
 \end{system}
\end{equation}
Since we have that dynamics of $\tilde{x}_1$ and $\delta_1$ are asymptotically stable, we just need to prove the stability of 
\begin{equation}
\begin{system}{cl}
j\omega \tilde{x}_2&=  (A-BK)\tilde{x}_2+ BK\delta_2\\
j\omega \delta_2&=\bar{d}_{22}A\delta_2
\end{system}
\end{equation}
Similar to the analysis of stability of system \eqref{x-1}, since $|\bar{d}_{22}|<1$, we have
\[
\delta_2\to 0, \text{ and } \tilde{x}_2\to 0,
\]
as $k \to \infty$. 
Similar to the case of $i=2$, we can obtain that \eqref{x-i}, for $i=3, \hdots, N$ and $j\le i-1$ is asymptotically stable, i.e. 
\[
\delta_i\to 0, \text{ and } \tilde{x}_i\to 0, \text{ as } k\to \infty.
\]
since we have $|\bar{d}_{ii}|<1$ and dynamics of $\tilde{x}_{i-1}$ and $\delta_{i-1}$ are asymptotically stable. Therefore we obtain that 
\[
\tilde{x}_i\to 0 \text{ as } k\to \infty
\]
for $i=2,\hdots, N$, which is equivalent to the stability of matrix \eqref{cond-matrix}. Then condition \eqref{cond-full} is satisfied. Therefore, based on \cite[Lemma 6]{zhang-saberi-stoorvogel-continues-discrete}, for all ${\kappa}_{ij}$ and $\hat{\kappa}_{ij}$,
\[
\bar{x}_i\to {x}_r
\]
as $k\to\infty$, which equivalently means that delayed synchronization \eqref{delayed-x-sync} is achieved.
\end{itemize}
\end{proof}

\subsubsection{\textbf{Partial-state coupling}}
In this subsection we consider MAS with partial-state coupling, i.e. $C\ne I$. \\

\begin{table}[h]
	\centering
	\captionsetup[table]{labelformat=empty}
	\caption*{Protocol 2: Partial-state Coupling}
	\begin{tabular}{p{15cm}}
		\toprule
		We design collaborative protocols based on localized information exchanges for agents $i=1,\hdots, N$ as
	\begin{equation}\label{pscps}
	\begin{system}{cl}
	{\hat{x}}_i(k+1)&=A\hat{x}_i(k)-BK\hat{\zeta}_i(k)+H(\bar{\zeta}_i(k)-C\hat{x}_i(k)),\\
	{\chi}_i(k+1)&=A\chi_i(k)+Bu_i(k)+A\hat{x}_i(k)-A\hat{\zeta}_i(k),\\
	u_i(k)&=-K\chi_i(k),
	\end{system}
	\end{equation}
		where $\bar{\zeta}_i(k)$ and $\hat{\zeta}_i(k)$ are defined by \eqref{zetabar2} and \eqref{add_1}. Matrix $K$ and $H$ are designed such that $A-BK$ and $A-HC$ are Schur stable.\\
		\bottomrule
\end{tabular}
\end{table} 
Then, we have the following theorem for scalable delayed regulated state synchronization of MAS with partial-state coupling.
\begin{theorem}\label{thm_p}
	Consider MAS \eqref{agent} consisting of $N$ agents satisfying Assumption \ref{agentass}. Let the associated network communication be given by \eqref{zetabar2}.  
	
	Then, the scalable delayed state synchronization problem as defined in Problem \ref{prob-x} is solvable. In particular, the linear dynamic protocol \eqref{pscps} solves delayed regulated state synchronization problem for any $N$ and any graph
	$\mathcal{G}\in\mathbb{G}^N$.
\end{theorem} 
\begin{proof}[Proof of Theorem \ref{thm_p}]
	Similar to the proof of Theorem \ref{thm_f} and by defining $\hat{\bar{x}}_i(k)=\hat{x}_i(k+\bar{\kappa}_{i,r})$ and $\hat{\bar{x}}=\begin{pmatrix}\hat{\bar{x}}_1\T,\hdots,\hat{\bar{x}}_N\T\end{pmatrix}\T$, we have the following closed-loop system in frequency domain as
	\[
	\begin{system}{cl}
	e^{j\omega}{\tilde{x}}&=(I\otimes A) \tilde{x}-(I\otimes BK)\bar{\chi}\\
	e^{j\omega}{\bar{\chi}}&=\left(I\otimes(A-BK)\right)\bar{\chi}+(I\otimes A)\hat{\bar{x}}-\left((I-\bar{D}_{j\omega}(\kappa))\otimes A\right)\bar{\chi}\\
		e^{j\omega}\hat{\bar{x}}&=\left(I\otimes(A-HC)\right)\hat{\bar{x}}-( (I-\bar{D}_{j\omega}(\kappa))\otimes BK)\bar{\chi}\\
	&\hspace{4.1cm}+((I-\bar{D})\otimes HC)\tilde{x}\end{system}
	\]
	then, by defining $\delta=\tilde{x}-\bar{\chi}$, and $\bar{\delta}=((I-\bar{D}_{j \omega}(\kappa))\otimes I)\tilde{x}-\hat{\bar{x}}$, we obtain
		\begin{equation}\label{error-partial}
	\begin{system}{cl}
	e^{j\omega}{\tilde{x}}&=(I\otimes A) \tilde{x}-(I\otimes BK)\bar{\chi}\\
	e^{j\omega}{\delta}&=(\bar{D}_{j\omega}(\kappa)\otimes A)\delta+(I\otimes A)\bar{\delta}\\
	e^{j\omega}{\bar{\delta}}&=\left(I\otimes(A-HC)\right)\bar{\delta}-((\bar{D}_{j\omega}(\kappa)-\bar{D})\otimes HC){\tilde{x}}
	\end{system}
	\end{equation}
	
	We prove \eqref{error-partial} is asymptotically stable for all communication
	delays ${\kappa}_{ij}\in \mathbb{N}^+$ and $\hat{\kappa}_{ij}\in \mathbb{N}^+$. Following the critical
	\cite[Lemma 6]{zhang-saberi-stoorvogel-continues-discrete}, we first prove stability without
	communication delays ${\kappa}_{ij}$ and $\hat{\kappa}_{ij}$ and then prove stability in presence of communication delays.
	\begin{itemize}
		\item In the absence of communication delays in the network,
		the stability of system \eqref{error-partial} is equivalent to the stability of matrix
		\begin{equation}
		\begin{pmatrix}
		I\otimes (A-BK)&I\otimes BK&0\\
		0&\bar{D}\otimes A&I\otimes A\\
		0&0&I\otimes(A-HC)
		\end{pmatrix}
		\end{equation}
		similar to the proof of Theorem \ref{thm_f}, we have all eigenvalues of $\bar{D}\otimes A$ are inside the unit disc.
		Then, since we have that $A-BK$ and $A-HC$ are Schur stable, we obtain that 
		\[
		\lim_{k\to\infty}\tilde{x}\to 0.
		\]
		It implies that $\bar{x}_i\to {x}_r$.
		\item  In the presence of communication delay, the closed-loop system \eqref{error-partial} is asymptotically stable if
		\begin{equation}\label{cond-partial}
		\det\left[j\omega I-\begin{pmatrix}
		I\otimes (A-BK)&I\otimes BK&0\\
		0&\bar{D}_{j\omega}(\kappa)\otimes A&I\otimes A\\
		(\bar{D}-\bar{D}_{j\omega}(\kappa))\otimes HC)&0&I\otimes(A-HC)
		\end{pmatrix}\right]\ne 0
		\end{equation}
		for all $\omega \in \mathbb{R}$ and any communication delays ${\kappa}_{ij}\in \mathbb{N}^+$ and $\hat{\kappa}_{ij}\in \mathbb{N}^+$.
		Condition \eqref{cond-partial} is satisfied if matrix
		\begin{equation}\label{cond-matrix-2}
	\begin{pmatrix}
	I\otimes (A-BK)&I\otimes BK&0\\
	0&\bar{D}_{j\omega}(\kappa)\otimes A&I\otimes A\\
	(\bar{D}-\bar{D}_{j\omega}(\kappa))\otimes HC)&0&I\otimes(A-HC)
	\end{pmatrix}
		\end{equation}
		has no eigenvalues on the unit circle for all $\omega \in \mathbb{R}$. 
		
		Then, according to the structure of matrix $\bar{D}$, \eqref{error-partial} can be rewritten as
		\begin{equation}
		\begin{system}{cl}\label{y-2}
		e^{j\omega} \tilde{x}_1&= (A-BK)\tilde{x}_1+ BK\delta_1\\
		e^{j\omega}{\delta}_1&=\bar{d}_{11}A\delta_1+A\bar{\delta}_1\\
		e^{j\omega}\bar{\delta}_1&=\left(A-HC\right)\bar{\delta}_1
		\end{system}
		\end{equation}
		and
		\begin{equation}\label{y-i}
		\begin{system}{cl}
		e^{j\omega} \tilde{x}_i&=  (A-BK)\tilde{x}_i+ BK\delta_i\\
		e^{j\omega} \delta_i&=\bar{d}_{ii}A\delta_i+\sum_{j=1}^{i-1}\bar{d}_{ij}e^{j\omega({\kappa}_{ij}-\hat{\kappa}_{ij})}A\delta_j+A\bar{\delta}_i\\
		e^{j\omega}\bar{\delta}_i&=\left(A-HC\right)\bar{\delta}_i+\sum_{j=1}^{i-1}(1-e^{j\omega{({\kappa}_{ij}-\hat{\kappa}_{ij})}})\bar{d}_{ij}HC\tilde{x}_j
		\end{system}
		\end{equation}
		for $i=2,\hdots, N$ and $j\le i-1$. 
		
		Then for $i=1$, we have
		\[
		\bar{\delta}_1\to 0 \text{ as } k\to \infty
		\]
		since $A-HC$ is Schur stable. In the following, since $\bar{d}_{11}<1$, one can obtain that all eigenvalues of $\bar{d}_{11}A$ are inside the unit disc, that is 
		\[
		\delta_1\to 0 \text{ as } k\to \infty
		\]
		then, given that $A-BK$ is Schur stable, we have
		\[
		\tilde{x}_1\to 0 \text{ as } k\to \infty
		\]
		Therefore, the dynamics of $\tilde{x}_1$, $\delta_1$ and $\bar{\delta}_1$ are asymptotically stable.
		
		Then for $i=2$ and $j=1$, we have 
		\begin{equation}
		\begin{system}{cl}
		e^{j\omega} \tilde{x}_2&=  (A-BK)\tilde{x}_2+ BK\delta_2\\
		e^{j\omega} \delta_2&=\bar{d}_{22}A\delta_2+\bar{d}_{21}e^{j\omega({\kappa}_{21}-\hat{\kappa}_{21})}A\delta_1+A\bar{\delta}_2\\
		e^{j\omega}\bar{\delta}_2&=(A-HC)\bar{\delta}_2+(1-e^{j\omega{({\kappa}_{21}-\hat{\kappa}_{21})}})\bar{d}_{21}HC\tilde{x}_1
		\end{system}
		\end{equation}
		Since we have that dynamics of $\tilde{x}_1$ and $\delta_1$ are asymptotically stable, we just need to prove the stability of 
		\begin{equation}
		\begin{system}{cl}
		e^{j\omega} \tilde{x}_2&=  (A-BK)\tilde{x}_2+ BK\delta_2\\
		e^{j\omega} \delta_2&=\bar{d}_{22}A\delta_2+A\bar{\delta}_2\\
		e^{j\omega}\bar{\delta}_2&=(A-HC)\bar{\delta}_2
		\end{system}
		\end{equation}
		Similar to the analysis of stability of system \eqref{y-2}, since $\bar{d}_{22}<1$, we have
		\[
		\delta_2\to 0,\text{ } \bar{\delta}_2\to 0, \text{ and } \tilde{x}_2\to 0,
		\]
		as $k \to \infty$. 
		Similar to the case of $i=2$, we can obtain that \eqref{y-i}, for $i=3, \hdots, N$ and $j\le i-1$ is asymptotically stable, i.e. 
		\[
		\delta_i\to 0, \text{ }\bar{\delta}_i\to 0, \text{ and } \tilde{x}_i\to 0, \text{ as } k\to \infty.
		\]
		since we have $\bar{d}_{ii}<1$ and dynamics of $\tilde{x}_{i-1}$ and $\delta_{i-1}$ are asymptotically stable. Therefore we obtain that 
		\[
		\tilde{x}_i\to 0 \text{ as } k\to \infty
		\]
		for $i=2,\hdots, N$, which is equivalent to the stability of system \eqref{error-partial}. Then condition \eqref{cond-partial} is satisfied. Therefore, based on \cite[Lemma 3]{zhang-saberi-stoorvogel-delay}, for all ${\kappa}_{ij}$ and $\hat{\kappa}_{ij}$,
		\[
		\bar{x}_i\to {x}_r
		\]
		as $k\to\infty$, which means that delayed synchronization \eqref{delayed-x-sync} is achieved.
	\end{itemize}
\end{proof}

\section{Heterogeneous MAS with introspective agents}\label{OS}
In this section, we study a heterogeneous MAS consisting of $N$ non-identical linear agents:
\begin{equation}\label{hete_sys}
	\begin{system*}{cl}
		x_i(k+1)&=A_ix_i(k)+B_iu_i(k),\\
		y_i(k)&=C_ix_i(k),
	\end{system*}
\end{equation}
where $x_i\in\mathbb{R}^{n_i}$, $u_i\in\mathbb{R}^{m_i}$ and $y_i\in\mathbb{R}^p$ are the state,
input, output of agent $i$
for $i=1,\ldots, N$.

The agents are introspective, meaning that each agent has access to its own local information. Specifically each agent has access to part of its state
\begin{equation}\label{local}
	z_i(k)=C_i^mx_i(k).
\end{equation}
where $z_i\in \mathbb{R}^{q_i}$.

The communication network provides agent $i$ with localized information \eqref{zetabar2} which is a linear combination of its own output relative to that of other agents. The agents have also access to the localized information defined by \eqref{info2}. We define the following definition.

\begin{definition}\label{def-synch-y}
	The agents of a heterogeneous MAS are said to achieve 
	\begin{itemize}
		\item delayed output synchronization if 
		\begin{equation}\label{delayed-y-sync}
			\lim_{k\to\infty}\left[(y_i(k)-y_j(k-{\kappa}_{ij})\right]=0, \quad \text{for all } i,j\in \{1,\hdots,N\}.
		\end{equation}
		where ${\kappa}_{ij}$ represents communication delay from agent $j$ to agent $i$.
		
		\item and delayed regulated output synchronization if 
		\begin{equation}\label{delayed-reg-sync-y}
			\lim_{k\to\infty}\left[(y_i(k)-y_r(k-{\kappa}_{ir})\right]=0, \quad \text{for all } i\in \{1,\hdots,N\}.
		\end{equation}
		where $\kappa_{ir}$ represents the sum of delays from agent $i$ to the exosystem.
	\end{itemize}
\end{definition}
We formulate the regulated output synchronization problem for heterogeneous network as follows.
\begin{problem}\label{prob_y}
	Consider a MAS \eqref{hete_sys} and \eqref{zetabar2}. Let
	$\mathbb{G}^N$ be the set of network graphs as defined in Definition \ref{ungrN}.
	
	The \textbf{scalable delayed regulated output synchronization problem based on localized information exchange utilizing collaborative protocols} for heterogeneous networks with unknown nonuniform and arbitrarily large communication delay is to find, if possible, a linear dynamic protocol for each agent $i\in\{1,...,N\}$, using only knowledge of agent models, i.e. $(C_i, A_i, B_i)$ of the form:
	\begin{equation}
		\begin{system}{cl}
			{x}_{i,c}(k+1)&=A_{i,c}x_{i,c}(k)+ B_{i,c}\bar{\zeta}_i(k)+ C_{i,c}\hat{\zeta}_i(k)+D_{i,c}z_i(k)\\
			u_i(k)&=E_{i,c}x_{i,c}(k)+F_{i,c}\bar{\zeta}_{i}(k)+G_{i,c}\hat{\zeta}_i(k)+H_{i,c}z_i(k),
		\end{system}
	\end{equation}
	where $\hat{\zeta}_i$ is defined by \eqref{info2} with $\xi_i=H_c x_{c,i}$ and $x_{c,i}\in \mathbb{R}^{n_c}$
	such that for any $N$, any graph $\mathscr{G}\in \mathbb{G}^N$, any communication delays $\kappa_{ij}$ and $\hat{\kappa}_{ij}$ we achieve delayed regulated output synchronization as stated by \eqref{delayed-reg-sync-y} in Definition \ref{def-synch-y}.
\end{problem}
We make the following assumptions on agents and the exosystem.
\begin{assumption}\label{ass2} For agents $i \in \{1,..., N\}$,
	\begin{enumerate}
		\item $(C_i,A_i,B_i)$ is stabilizable and detectable.
		\item $(C_i,A_i,B_i)$ is right-invertible.
		\item $(C_i^m, A_i)$ is detectable.
	\end{enumerate}
	
	\begin{assumption}\label{ass-exo}
		For exosystem,
		\begin{enumerate}
			\item $(C_r, A_r)$ is observable.
			\item  All the eigenvalues of $A_r$ are on the unit circle.
		\end{enumerate}
	\end{assumption}
\end{assumption}
We design scale-free protocols to solve scalable delayed regulated output synchronization problem as stated in Problem \ref{prob_y}. After introducing the architecture of our protocol, we design the protocols through four steps.

\subsection{{Architecture of the protocol}}
Our protocol has the structure shown below in Figure \ref{Heterogeneous_reg}.

\begin{figure}[h]
	\includegraphics[width=10cm, height=6cm]{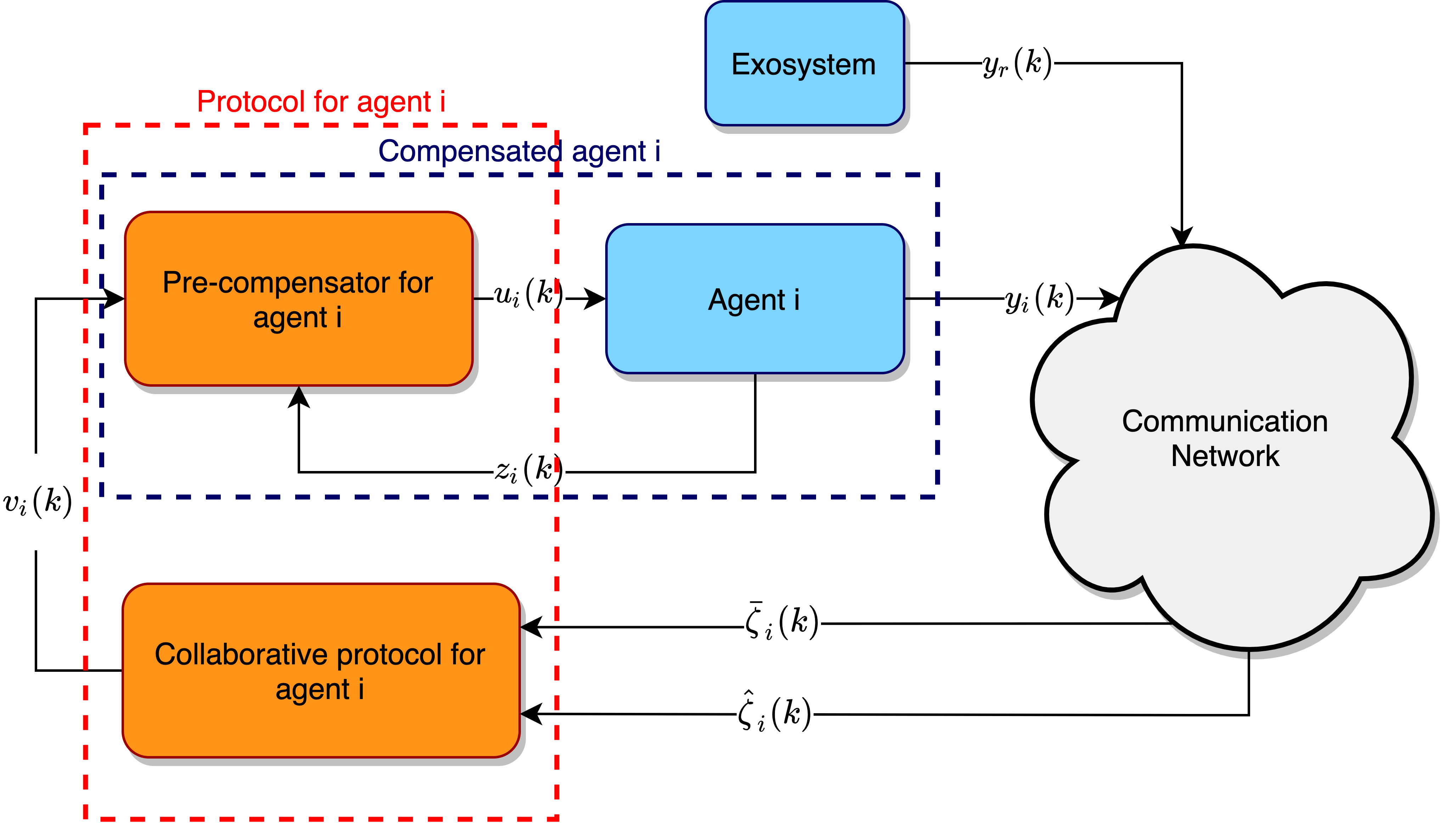}
	\centering
	\caption{Architecture of scale-free protocols for output synchronization of heterogeneous networks }\label{Heterogeneous_reg}
	\vspace*{-.3cm}
\end{figure}
As seen in the figure, our design methodology consists of two major modules.

\begin{enumerate}
	\item The first module is remodeling the exosystem to obtain the target model by designing pre-compensators following our previous results in \cite{yang-saberi-stoorvogel-grip-journal}.
	\item The second module is designing collaborate protocols for almost homogenized agents to achieve output and regulated output synchronization.
\end{enumerate}
\subsection{Protocol design}
To design our protocols, first we recall the following Lemma.
\begin{lemma}[\cite{yang-saberi-stoorvogel-grip-journal}]\label{lem-exo} There exists another exosystem given by:
	\begin{equation}\label{exo-2}
		\begin{system*}{cl}
			{\check{x}}_r(k+1)&=\check{A}_r\check{x}_r(k), \quad \check{x}_r(0)=\check{x}_{r0}\\
			y_r(k)&=\check{C}_r\check{x}_r(k),
		\end{system*}
	\end{equation}
	such that for all $x_{r0} \in \mathbb{R}^r$, there exists $\check{x}_{r0}\in \mathbb{R}^{\check{r}}$ for which \eqref{exo-2} generate exactly the same output $y_r$ as the original exosystem \eqref{exo}. Furthermore, we can find a matrix $\check{B}_r$ such that the triple $(\check{C}_r,\check{A}_r,\check{B}_r)$ is invertible, of uniform rank $n_q$, and has no invariant zero, where $n_q$ is an integer greater than or equal to maximal order of infinite zeros of $(C_i,A_i,B_i), i\in \{1,...,N\}$ and all the observability indices of $(C_r, A_r)$. Note that the eigenvalues of $\check{A}_r$ consists of all eigenvalues of $A_r$ and additional zero eigenvalues. 
\end{lemma}
We design our protocols through the four steps in \emph{Protocol 3}.
\begin{table}[h!]
	\centering
	\captionsetup[table]{labelformat=empty}
	\caption*{\textbf{Protocol 3.} Heterogeneous MAS}
	\begin{tabular}{p{15cm}}
		\toprule
		\textbf{Step 1: remodeling the exosystem} First, we remodel the exosystem to arrive at suitable choice for the target model $(\check{C}_r,\check{A}_r,\check{B}_r)$ following the design procedure in \cite{yang-saberi-stoorvogel-grip-journal} summarized in Lemma \ref{lem-exo}.
		
		\textbf{Step 2: designing pre-compensators} In this step, given the target model $(\check{C}_r,\check{A}_r,\check{B}_r)$, by utilizing the design methodology from \cite[Appendix B]{yang-saberi-stoorvogel-grip-journal}, we design a pre-compensators for each agent $i \in \{1, \dots, N\}$ of the form
		
		\begin{equation}\label{pre_con}
			\begin{system}{cl}
				{\xi}_i(k+1)&=A_{i,h}\xi_i(k)+B_{i,h}z_i(k)+E_{i,h}v_i(k),\\
				u_i(k)&=C_{i,h}\xi_i(k)+D_{i,h}v_i(k),
			\end{system}
		\end{equation} 
		which yields the compensated agents as
		\begin{equation}\label{sys_homo}
			\begin{system*}{cl}
				{{x}}^h_i(k+1)&=\check{A}_r{x}^h_i(k)+\check{B}_r(v_i(k)+\rho_i(k)),\\
				{y}_i(k)&=\check{C}_rx^h_i(k),
			\end{system*}
		\end{equation} 
		where $\rho_i(k)$ is given by 
		\begin{equation}\label{sys-rho}
			\begin{system*}{cl}
				{\omega}_i(k+1)&=A_{i,s}\omega_i(k),\\
				\rho_i(k)&=C_{i,s}\omega_i(k),
			\end{system*}
		\end{equation}
		and $A_{i,s}$ is Schur stable. Note that the compensated agents are homogenized and have the target model $(\check{C}_r, \check{A}_r, \check{B}_r)$.

		\textbf{Step 3: designing collaborative protocols for the compensated agents} Collaborative protocols based on localized information exchanges are designed for the compensated agents $i=1,\hdots, N$ as
		
		\begin{equation}\label{pscp2}
		\begin{system}{cl}
{\hat{x}}_i(k+1)&=\check{A}_r\hat{x}_i(k)-\check{B}_rK\hat{\zeta}_i(k)+H(\bar{\zeta}_i(k)-\check{C}_r\hat{x}_i(k)),\\
{\chi}_i(k+1)&=\check{A}_r\chi_i(k)+\check{B}_rv_i(k)+\check{A}_r\hat{x}_i(k)-\check{A}_r\hat{\zeta}_i(k),\\
v_i(k)&=-K\chi_i(k),
\end{system}
		\end{equation}
		where $H$ and $K$ are matrices such that $\check{A}_r-H\check{C}_r$ and $\check{A}_r-\check{B}_rK$ are Schur stable.
		The exchanging information $\hat{\zeta}_i$ is defined as \eqref{info2} and $\bar{\zeta}_i$ is defined as \eqref{zetabar2}. 
		
		\textbf{Step 4: obtaining the protocols}
		The final protocol which is the combination of module $1$ and $2$ is
		\begin{equation}\label{pscp2final}
			\begin{system}{cl}
				{\xi}_i(k+1)&=A_{i,h}\xi_i(k)+B_{i,h}z_i(k)-E_{i,h}K\chi_i(k),\\
				{\hat{x}}_i(k+1)&=\check{A}_r\hat{x}_i(k)-\check{B}_rK\hat{\zeta}_i(k)+H(\bar{\zeta}_i-\check{C}_r\hat{x}_i(k)),\\
				{\chi}_i(k+1)&=(\check{A}_r-\check{B}_rK)\chi_i(k)+\check{A}_r\hat{x}_i(k)-\check{A}_r\hat{\zeta}_i(k),\\
				u_i(k)&=C_{i,h}\xi_i(k)-D_{i,h}K\chi_i(k),
			\end{system}
		\end{equation} \\
		\bottomrule
	\end{tabular}
\end{table}

Then, we have the following theorem for scalable regulated output synchronization of heterogeneous MAS.

\begin{theorem}\label{thm_reg_out_syn}
	Consider a heterogeneous network of $N$ agents \eqref{hete_sys} satisfying Assumption \ref{ass2} with local information \eqref{local} and the associated exosystem \eqref{exo} satisfying Assumption \ref{ass-exo}. Then, the scalable delayed regulated output synchronization problem as defined in Problem \ref{prob_y} is solvable.  In particular, the dynamic protocol \eqref{pscp2final} solves the scalable delayed regulated output synchronization problem based on localized information exchange for any $N$ and any graph
	$\mathscr{G}\in\mathbb{G}^N_\mathscr{C}$. 
\end{theorem}

\begin{proof}[Proof of Theorem \ref{thm_reg_out_syn}]
	Similar to the proof of Theorem \ref{thm_p} and by defining $\bar{x}_i(k)=x_i^h(k+\kappa_{ir})$, $\bar{\rho}_i(k)=\rho_i(k+\kappa_{ir})$, $\bar{\omega}_i(k)=\omega_i(k+\kappa_{ir})$, $\tilde{x}_i(k)=\bar{x}_i(k)-\check{x}_r(k)$ and 
	\begin{equation*}
		\tilde{x}=\begin{pmatrix}
			\tilde{x}_1\\ \vdots\\ \tilde{x}_N
		\end{pmatrix},\hat{\bar{x}}=\begin{pmatrix}
			\hat{\bar{x}}_1\\ \vdots\\ \hat{\bar{x}}_N
		\end{pmatrix},\bar{\chi}=\begin{pmatrix}
			\bar{\chi}_1\\ \vdots\\ \bar{\chi}_N
		\end{pmatrix},\bar{\rho}=\begin{pmatrix}
			\bar{\rho}_1\\ \vdots\\ \bar{\rho}_N\end{pmatrix},\bar{\omega}=\begin{pmatrix}
			\bar{\omega}_1\\ \vdots\\\bar{\omega}_N\end{pmatrix}
	\end{equation*}
	then, we have the following closed-loop system in frequency domain
	\begin{equation}
		\begin{system}{cl}
			e^{j\omega}{\tilde{x}}&=(I\otimes \check{A}_r)\tilde{x}-(I\otimes \check{B}_rK)\bar{\chi}+(I\otimes \check{B}_r)\bar{\rho}\\
			e^{j\omega}{\hat{\bar{x}}}&=(I\otimes (\check{A}_r-H\check{C}_r))\hat{\bar{x}}-((I-\bar{D}_{j\omega}(\kappa))\otimes \check{B}_rK)\bar{\chi}+((I-\bar{D})\otimes H\check{C}_r)\tilde{x}\\
			e^{j\omega}{\bar{\chi}}&=(I\otimes(\check{A}_r-\check{B}_rK))\bar{\chi}-((I-\bar{D}_{j\omega}(\kappa))\otimes \check{A}_r)\bar{\chi}+(I\otimes \check{A}_r)\hat{\bar{x}}\\
			e^{j\omega}\omega&=A_s\omega
		\end{system}
	\end{equation}
	
	By defining $\delta=\tilde{x}-\bar{\chi}$ and $\bar{\delta}=((I-\bar{D}_{j\omega}(\kappa))\otimes I)\tilde{x}-\hat{\bar{x}}$, we obtain  
	\begin{equation}\label{newsystem3}
		\begin{system*}{cl}
			e^{j\omega}{\tilde{x}}&=(I\otimes (\check{A}_r-\check{B}_rK ))\tilde{x}+(I\otimes \check{B}_rK )\delta+(I\otimes \check{B}_r)C_s\bar{\omega}\\
			e^{j\omega}{\delta}&=(\bar{D}_{j\omega}(\kappa)\otimes \check{A}_r)\delta+(I\otimes \check{A}_r)\bar{\delta}+(I\otimes \check{B}_r)C_s\bar{\omega}\\
			e^{j\omega}{\bar{\delta}}&=(I\otimes(\check{A}_r-H\check{C}_r))\bar{\delta}-((\bar{D}_{j\omega}(\kappa)-\bar{D})\otimes H\check{C}_r)\tilde{x}+((I-\bar{D}_{j\omega}(\kappa))\otimes \check{B}_r)C_s\bar{\omega}\\e^{j\omega}\omega&=A_s\omega
		\end{system*}
	\end{equation}
	
	where $C_s=\diag\{C_{i,s}\}$ for $i=1,...,N$. Similar to the proof of Theorem \ref{thm_p}, we prove \eqref{newsystem3} is asymptotically stable for all communication
	delays ${\kappa}_{ij}\in \mathbb{R}^+$ and $\hat{\kappa}_{ij}\in \mathbb{R}^+$. Following the critical
	lemma \cite[Lemma 3]{zhang-saberi-stoorvogel-delay}, we first prove stability without
	communication delays ${\kappa}_{ij}$ and $\hat{\kappa}_{ij}$ and then prove stability in presence of communication delays.
	\begin{itemize}
		\item In the absence of communication delays in the network,
		the stability of system \eqref{newsystem3} is equivalent to the stability of matrix
		
		
		\begin{equation}
			\begin{pmatrix}
				I\otimes (\check{A}_r-\check{B}_rK)&I\otimes \check{B}_rK&0&(I\otimes \check{B}_r)C_s\\
				0&\bar{D}\otimes \check{A}_r&I\otimes \check{A}_r&(I\otimes \check{B}_r)C_s\\
				0&0&I\otimes(\check{A}_r-H\check{C}_r)&((I-\bar{D})\otimes \check{B}_r)C_s\\
				0&0&0&A_s
			\end{pmatrix}
		\end{equation}
		where $A_s=\diag\{A_{i,s}\}$ for $i=1,...,N$. Similar to the proof of Theorem \ref{thm_p}, we have that all eigenvalues of $\bar{D}\otimes \check{A}_r$ are inside the unit disc.
		Then, since we have that $\check{A}_r-\check{B}_rK$,  $\check{A}_r-H\check{C}_r$ and $A_s$ are Schur stable, we obtain that 
		\[
		\lim_{k\to\infty}\tilde{x}\to 0.
		\]
		It implies that $\bar{x}_i\to {x}_r$.
		
		\item  In the presence of communication delay, the closed-loop system \eqref{newsystem3} is asymptotically stable if
		
		\begin{equation}\label{cond-partial-h}
			\det\left[e^{j\omega} I-	\begin{pmatrix}
			I\otimes (\check{A}_r-\check{B}_rK)&I\otimes \check{B}_rK&0&(I\otimes \check{B}_r)C_s\\
			0&\bar{D}_{j\omega}(\kappa)\otimes \check{A}_r&I\otimes \check{A}_r&(I\otimes \check{B}_r)C_s\\
			(\bar{D}-\bar{D}_{j\omega}(\kappa))\otimes H\check{C}_r&0&I\otimes(\check{A}_r-H\check{C}_r)&((I-\bar{D}_{j\omega}(\kappa))\otimes \check{B}_r)C_s\\
			0&0&0&A_s
			\end{pmatrix}\right]\ne0
		\end{equation}
		
		for all $\omega \in \mathbb{R}$ and any communication delays ${\kappa}_{ij}\in \mathbb{R}^+$ and $\hat{\kappa}_{ij}\in \mathbb{R}^+$.
		Condition \eqref{cond-partial-h} is satisfied if matrix
		
		\begin{equation}
		\begin{pmatrix}
		I\otimes (\check{A}_r-\check{B}_rK)&I\otimes \check{B}_rK&0&(I\otimes \check{B}_r)C_s\\
		0&\bar{D}_{j\omega}(\kappa)\otimes \check{A}_r&I\otimes \check{A}_r&(I\otimes \check{B}_r)C_s\\
		(\bar{D}-\bar{D}_{j\omega}(\kappa))\otimes H\check{C}_r&0&I\otimes(\check{A}_r-H\check{C}_r)&((I-\bar{D}_{j\omega}(\kappa))\otimes \check{B}_r)C_s\\
		0&0&0&A_s
		\end{pmatrix}
		\end{equation}
		has no eigenvalues on the unit circle for all $\omega \in \mathbb{R}$. Then, according to the structure of  matrix $\bar{D}$, and similar to the proof of Theorem \ref{thm_p} one can obtain that $\tilde{x}$ is asymptotically stable, i.e., $\lim_{k\to\infty}\tilde{x}_i=0$, which implies that $\lim_{k\to\infty}\tilde{y}_i=0$, or $\bar{y}_i\to y_r$.
	\end{itemize}
\end{proof}

\begin{figure}[t]
	\centering
	\begin{minipage}{.5\textwidth}
		\centering
		\includegraphics[width=4.7cm,height=3cm]{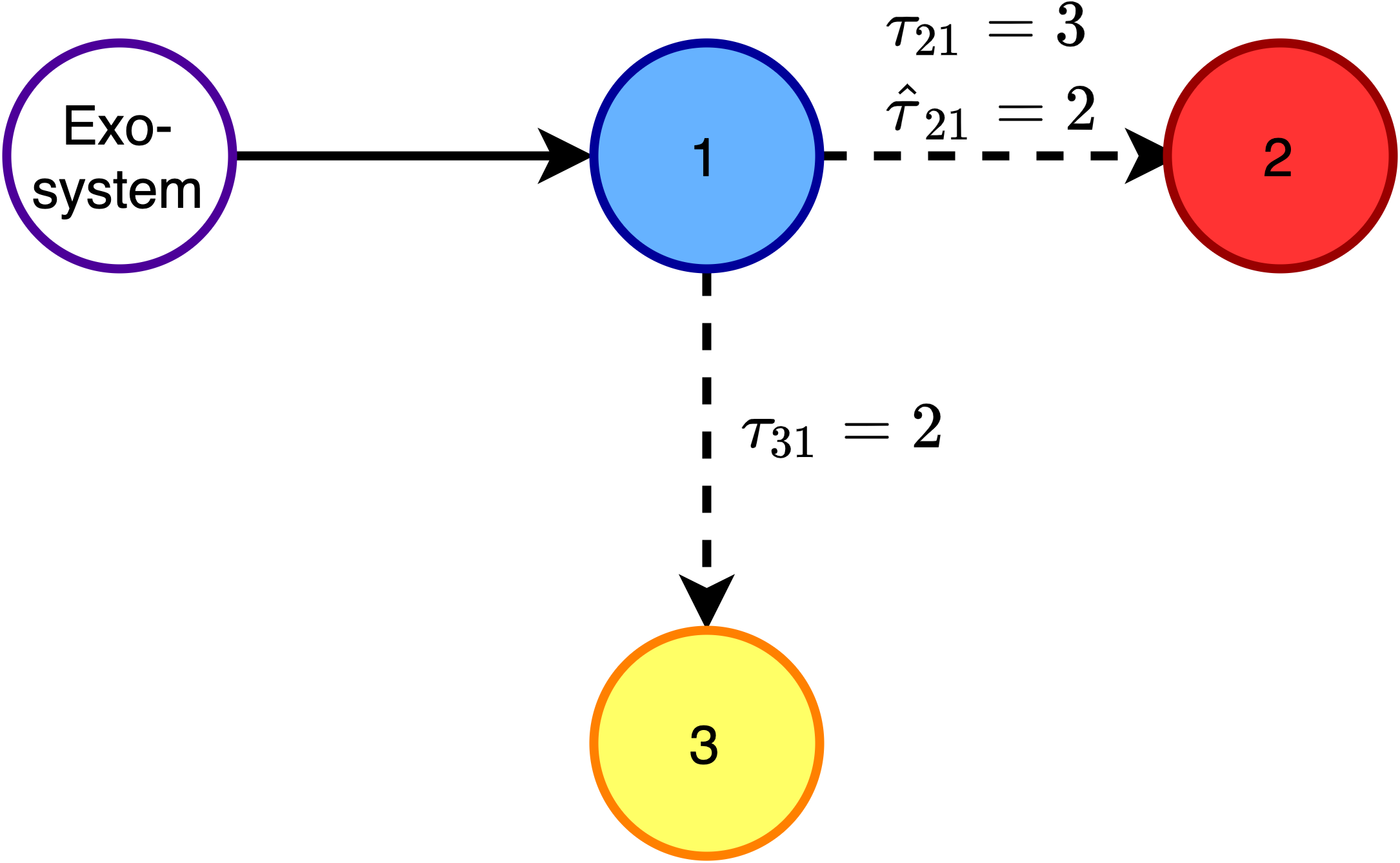}
		\caption{Communication graph of the network with $3$ nodes}\label{Tree-3nodes}
	\end{minipage}%
	\begin{minipage}{.5\textwidth}
		\centering
		\includegraphics[width=5.5cm,height=5.5cm]{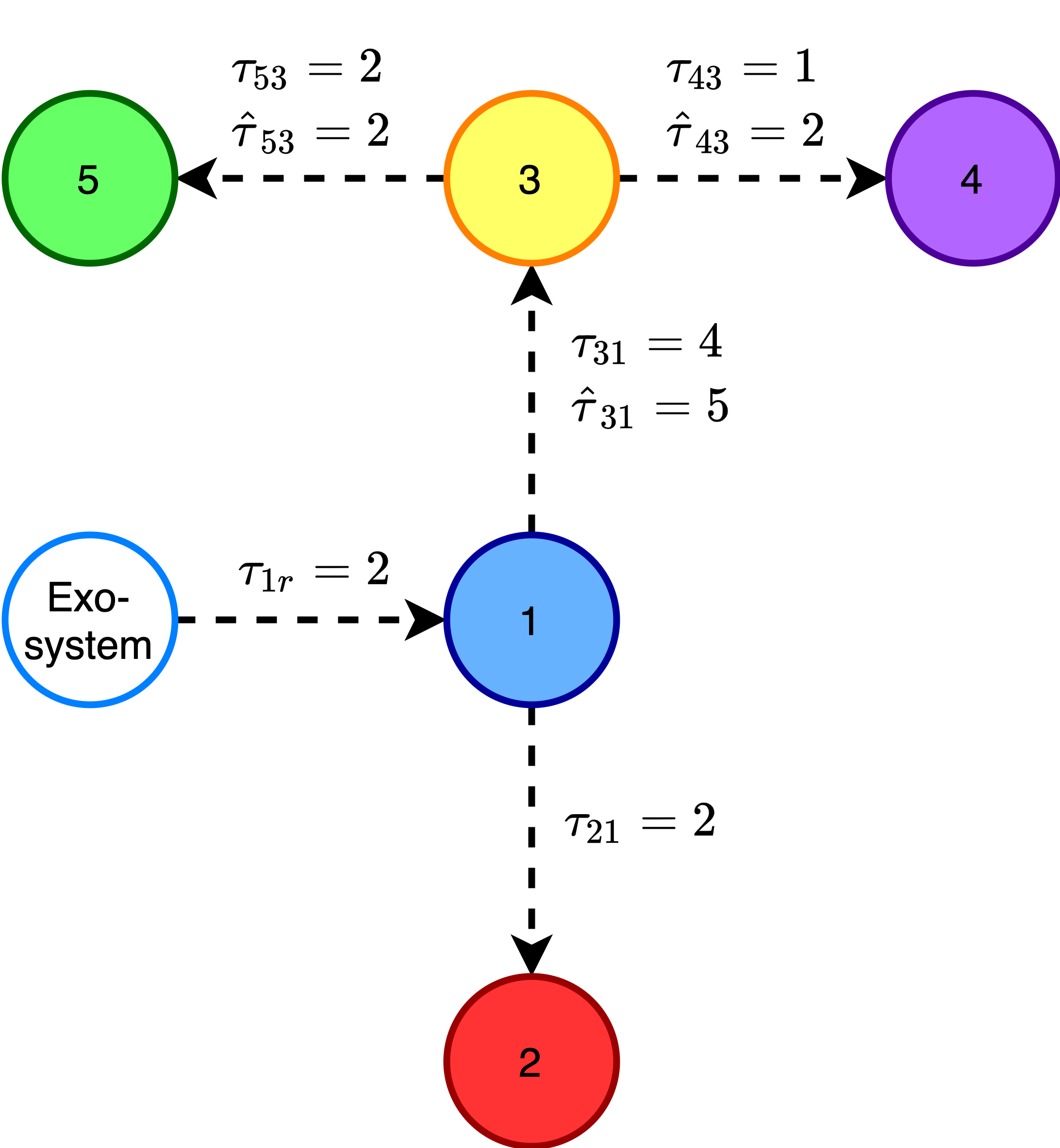}
		\caption{Communication graph of the network with $5$ nodes}\label{Tree-5nodes}
	\end{minipage}
\end{figure}
\begin{figure}[t!]
	\includegraphics[width=7cm, height=7cm]{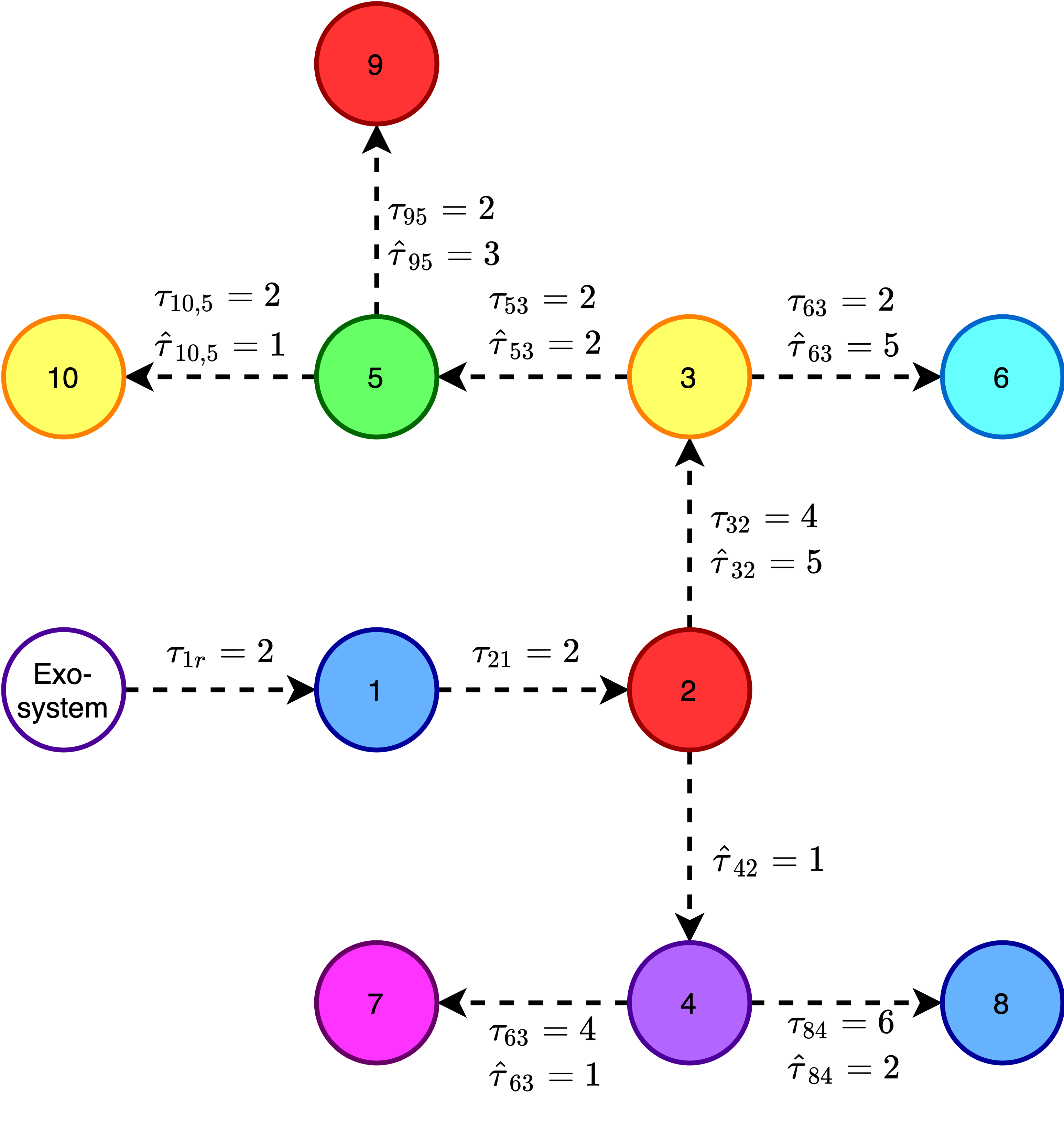}
	\centering
	\caption{Communication graph of the network with $10$ nodes}\label{Tree-10nodes}
\end{figure} 
\section{Numerical Example}
In this section, we will illustrate the performance of our scalable linear protocols with numerical examples for scale-free delayed regulated state synchronization of homogeneous MAS with partial-state coupling and scale-free delayed regulated output synchronization for heterogeneous MAS in presence of communication delays.

\subsection*{Example 1. Homogeneous MAS}
Consider agents models \eqref{agent} with matrices
\begin{equation*}
A=\begin{pmatrix}
0.5 &1 &1\\
0 &\sqrt{3}/2 &-0.5\\
0 &0.5 &\sqrt{3}/2
\end{pmatrix},\quad B=\begin{pmatrix}
1\\1\\0
\end{pmatrix}, \quad C=\begin{pmatrix}
1&0&1
\end{pmatrix}.\\
\end{equation*} 

The goal is delayed regulated state synchronization to a dynamic synchronized trajectory generated by 
\[
\begin{system}{cl}
\dot{x}_r&=\begin{pmatrix}
0.5 &1 &1\\
0 &\sqrt{3}/2 &-0.5\\
0 &0.5 &\sqrt{3}/2
\end{pmatrix}x_r\\
y_r&=\begin{pmatrix}
1&0&1
\end{pmatrix}x_r
\end{system}
\]
by choosing initial condition, $x_r(0)=\begin{pmatrix}0.3&0.1&0.1
\end{pmatrix}\T$. Meanwhile, to show the scalability of our protocols we choose three different MAS, with different communication networks and different number of agents. In all of the following cases we choose matrices $K=\begin{pmatrix}0.0695&1.7625&1.2051\end{pmatrix}$ and $H=\begin{pmatrix}1.4327&0.4143&0.6993\end{pmatrix}\T$.

\textit{Case I}: Consider a MAS consisting of $3$ agents with agent models $(A,B,C)$ and tree communication topology shown in Figure \ref{Tree-3nodes}. The communication delays are equal to $\kappa_{21}=3$, $\kappa_{31}=2$, and $\hat{\kappa}_{21}=2$. The exosystem provides $x_r(t)$ for agent $1$. Figure \ref{Hom-3-dis} shows the simulation results.

\textit{Case II}: Now, we consider another MAS consisting of $5$ agents with agent models $(A,B,C)$ and communication topology shown in Figure \ref{Tree-5nodes}. The communication delays are equal to $\kappa_{1r}=2$, $\kappa_{21}=2$, $\kappa_{31}=4$, $\kappa_{43}=1$, $\kappa_{53}=2$, $\hat{\kappa}_{31}=5$, $\hat{\kappa}_{43}=2$ and $\hat{\kappa}_{53}=2$. We show that with the same protocol utilized for \emph{case I}, we achieve delayed regulated state synchronization. The simulation results are shown in Figure \ref{Hom-5-dis}.

\textit{Case III}: Finally, consider a MAS consisting of $10$ agents with agent models $(A,B,C)$ and directed communication topology shown in Figure \ref{Tree-10nodes}. The communication delays are equal to $\kappa_{1r}=2$, $\kappa_{21}=2$, $\kappa_{32}=4$, $\kappa_{53}=2$, $\kappa_{63}=2$, $\kappa_{74}=4$, $\kappa_{84}=6$, $\kappa_{95}=2$, $\kappa_{10,5}=2$, $\hat{\kappa}_{32}=5$, $\hat{\kappa}_{42}=1$, $\hat{\kappa}_{53}=2$, $\hat{\kappa}_{63}=5$, $\hat{\kappa}_{74}=1$, $\hat{\kappa}_{84}=6$, $\hat{\kappa}_{95}=3$, and $\hat{\kappa}_{10,5}=1$. The exosystem provides $x_r$ for agent $1$. The simulation results for this MAS are presented in Figure \ref{Hom-10-dis}.

\begin{figure}[t]
	\includegraphics[width=13cm, height=9cm]{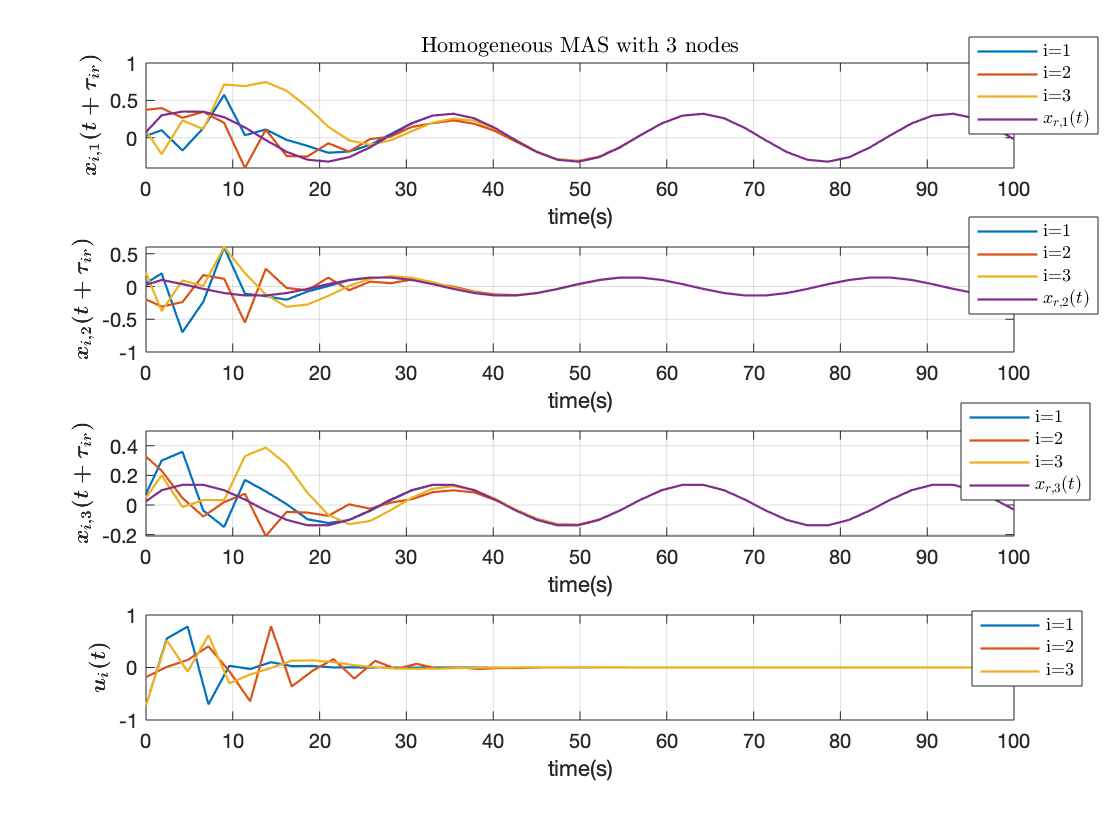}
	\centering
	\caption{Scale-free delayed regulated state synchronization for homogeneous MAS with $3$ nodes}\label{Hom-3-dis}
\end{figure} 
\begin{figure}[t!]
	\includegraphics[width=13cm, height=9cm]{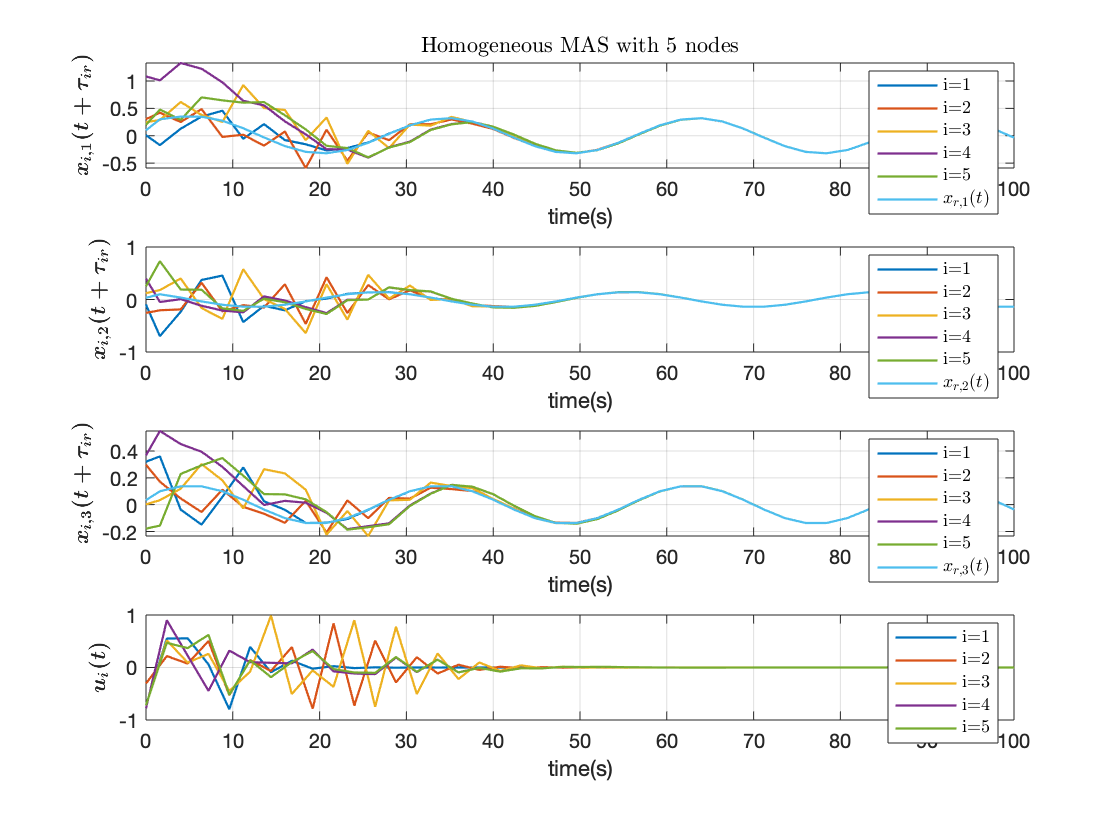}
	\centering
	\caption{Scale-free delayed regulated state synchronization for homogeneous MAS with $5$ nodes}\label{Hom-5-dis}
\end{figure} 
\begin{figure}[t!]
	\includegraphics[width=13cm, height=9cm]{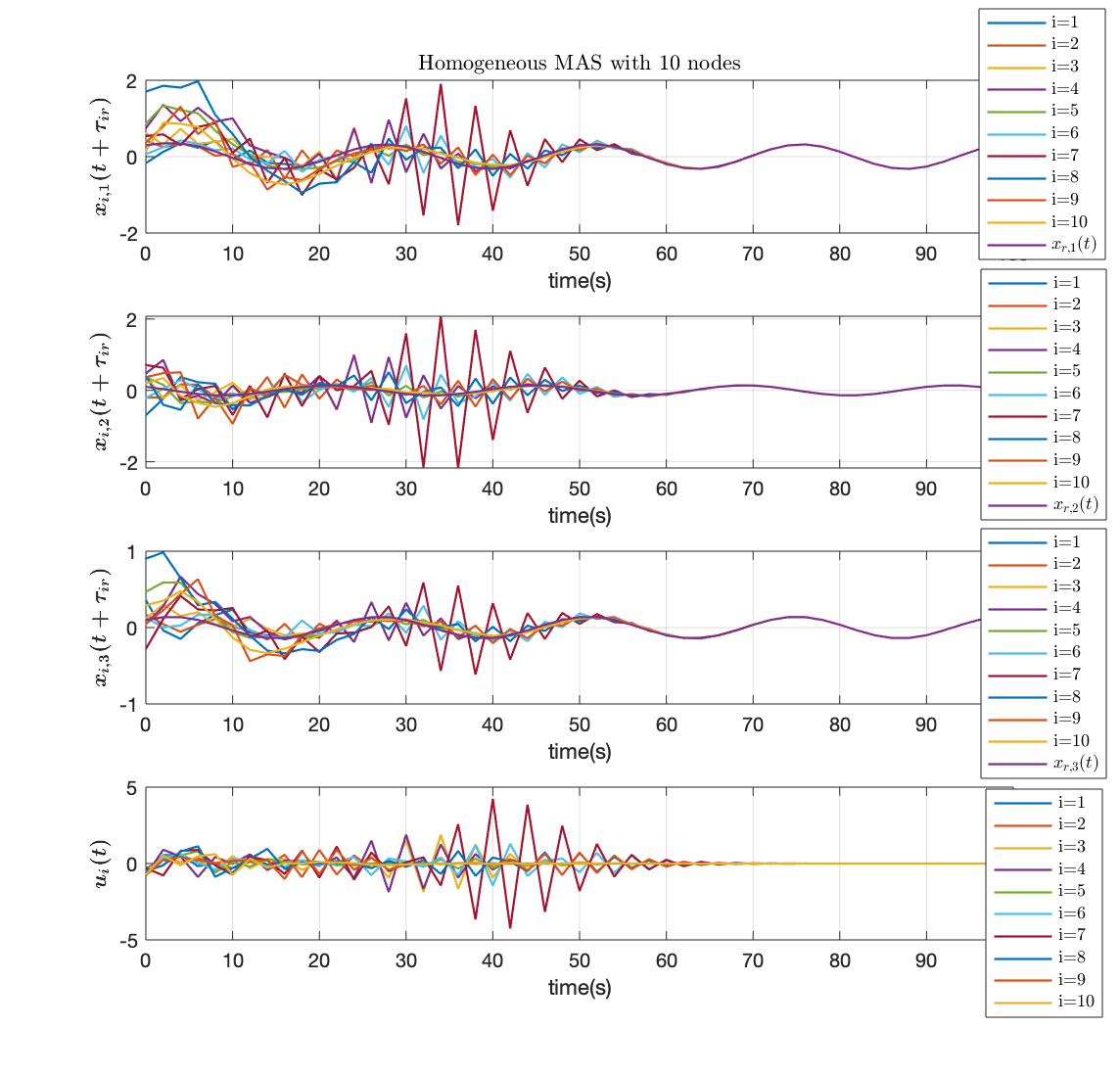}
	\centering
	\caption{Scale-free delayed regulated state synchronization for homogeneous MAS with $10$ nodes}\label{Hom-10-dis}
\end{figure} 
The simulation results show that our one-shot-design protocol \eqref{pscps} achieves delayed regulated state synchronization for any communication network with associated spanning tree graph and any size of the network. Moreover, the protocol can tolerate any unknown non-uniform and arbitrarily large communication delays.
\begin{figure}[t]
	\includegraphics[width=13cm, height=9cm]{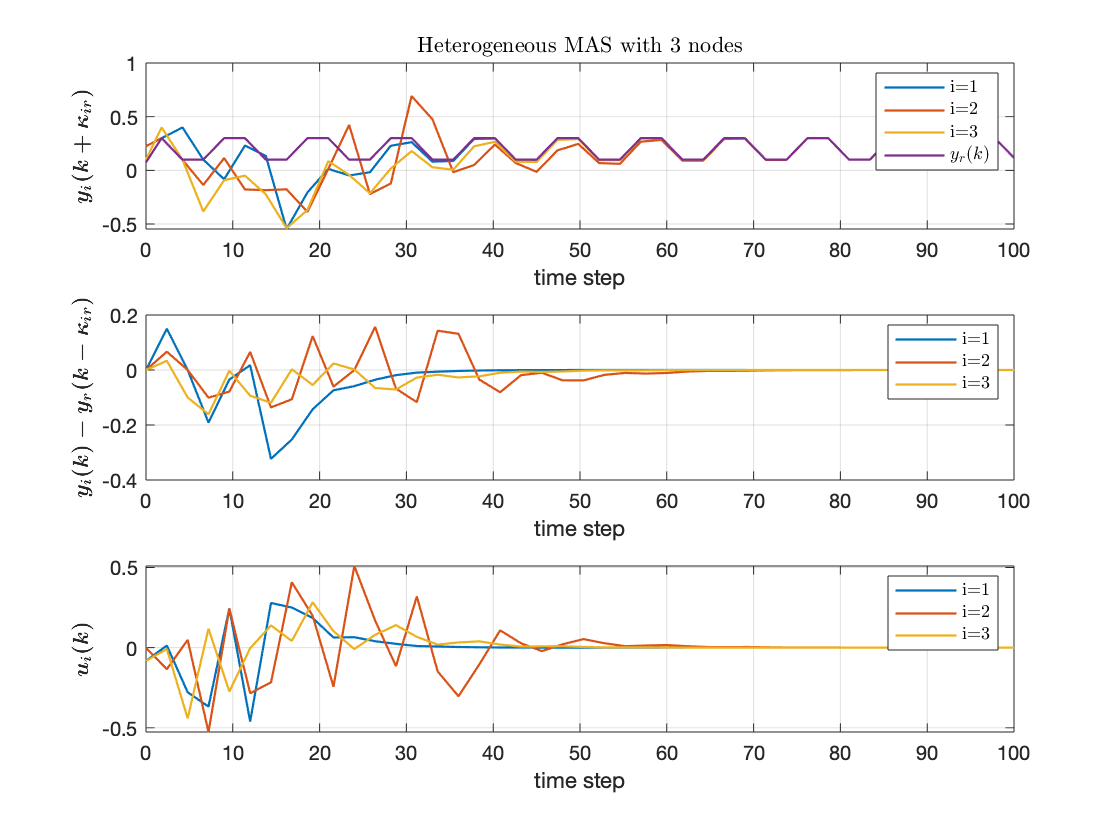}
	\centering
	\caption{Scale-free delayed regulated output synchronization for heterogeneous MAS with $3$ nodes}\label{Het-3-dis}
\end{figure} 
\begin{figure}[t!]
	\includegraphics[width=13cm, height=9cm]{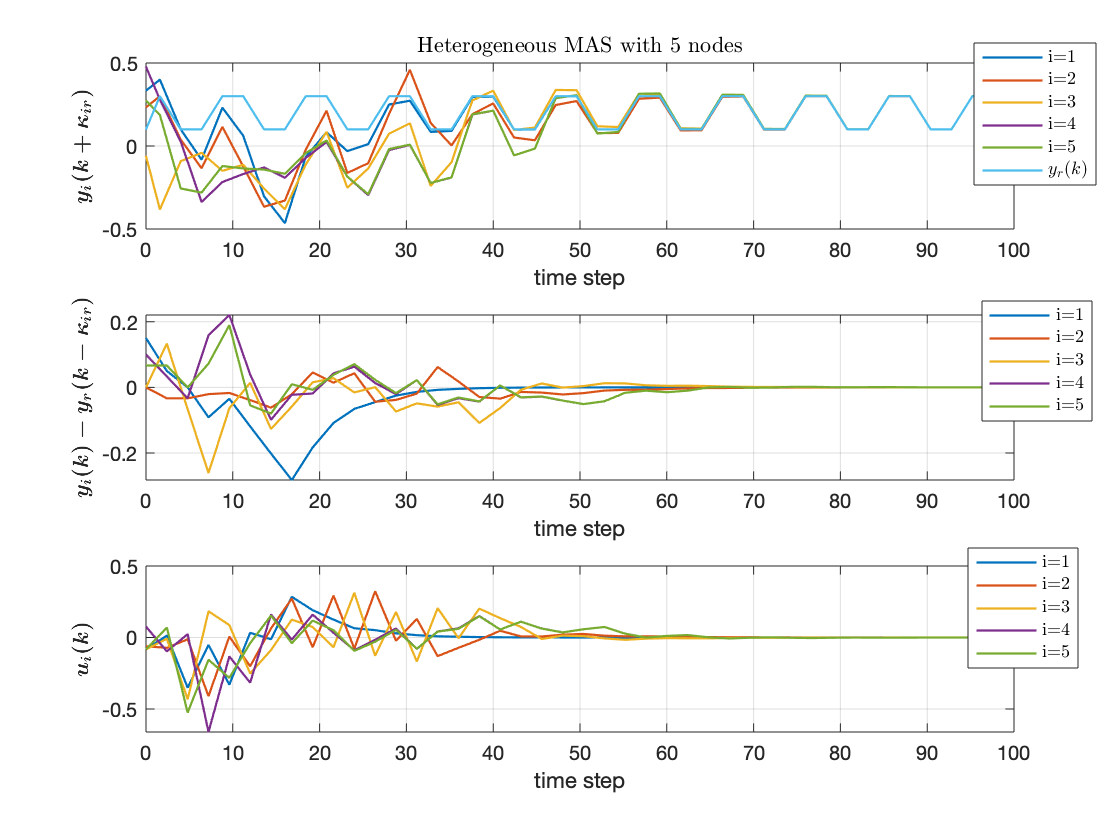}
	\centering
	\caption{Scale-free delayed regulated output synchronization for heterogeneous MAS with $5$ nodes}\label{Het-5-dis}
\end{figure} 
\begin{figure}[t!]
	\includegraphics[width=13cm, height=9cm]{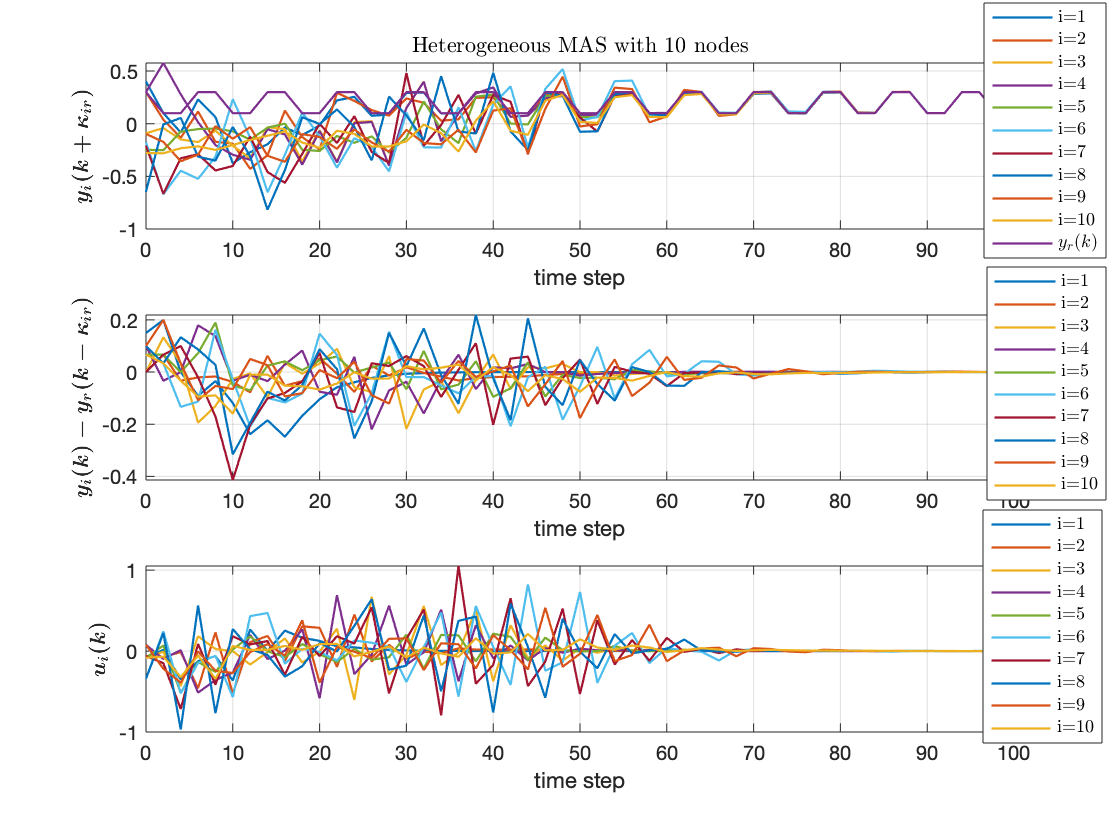}
	\centering
	\caption{Scale-free delayed regulated output synchronization for heterogeneous MAS with $10$ nodes}\label{Het-10-dis}
\end{figure} 
\subsection*{Example 2. Heterogeneous MAS}
In this section, we will illustrate the effectiveness of our protocols with a numerical example for delayed regulated output synchronization of heterogeneous discrete-time MAS. We show that our protocol design \emph{Protocol 3} is scale-free and it works for any graph $\mathcal{G}\in \mathbb{G}^N$ with any number of agents. Consider the agents models \eqref{hete_sys} with
\begin{equation*}
A_i=\begin{pmatrix}
0 &0 &1 &0\\0 &0 &0 &1\\0 &2 &1&1\\1 &1 &1 &0
\end{pmatrix}, B_i=\begin{pmatrix}
0&0\\0&0\\1&0\\0&1
\end{pmatrix}, C_i=\begin{pmatrix}
0&0&1&0
\end{pmatrix},
C^m_i=I
\end{equation*}
for $i=1,6$, and 
\begin{equation*}
A_i=\begin{pmatrix}
0&1&0\\0&0&1\\0&0&0
\end{pmatrix},B_i=\begin{pmatrix}
0\\0\\1
\end{pmatrix},C_i=\begin{pmatrix}
1&0&0
\end{pmatrix},
C^m_i=I,
\end{equation*}
for $i=2,7$, and
\begin{equation*}
A_i=\begin{pmatrix}
0 &0&0&1 &0\\ 0&1 &-1 &0 &1\\ 0 &1 &0 &0 &0\\ 0& 0 &1 &0 &0\\ 1 &1& 0 &0 &1
\end{pmatrix},B_i=\begin{pmatrix}
0&0\\1&0\\0&0\\0&0\\0&1
\end{pmatrix},C_i=\begin{pmatrix}
0&0&1&0&0
\end{pmatrix},
C^m_i=I,
\end{equation*}
for $i=3,4,8,9$, and
\begin{equation*}
A_i=\begin{pmatrix}
0&1&0\\0&0&1\\-2&1&0
\end{pmatrix},B_i=\begin{pmatrix}
0\\0\\1
\end{pmatrix},C_i=\begin{pmatrix}
1&0&0
\end{pmatrix},
C^m_i=I,
\end{equation*}
for $i=5,10$. Note that $\bar{n}_d=3$, which is the degree of infinite zeros of $(C_2,A_2,B_2)$. In this example, our goal is delayed output regulation to a non-constant signal generated by
\begin{equation*}
\begin{system*}{cl}
\dot{x}_r&=\begin{pmatrix}
0&1&0\\0&0&1\\1&-1&1
\end{pmatrix}x_r,\\
y_r&=\begin{pmatrix}
1&0&0
\end{pmatrix}x_r.
\end{system*}
\end{equation*}
Utilizing Lemma \ref{lem-exo}, we choose $(\check{C}_r, \check{A}_r, \check{B}_r)$ as
\begin{equation*} 
\check{A}_r=\begin{pmatrix}
0&1&0\\0&0&1\\1&-1&1
\end{pmatrix},\quad \check{B}_r=\begin{pmatrix}
0\\0\\1
\end{pmatrix}, \quad \check{C}_r=\begin{pmatrix}
1&0&0
\end{pmatrix}
\end{equation*}
and $K=\begin{pmatrix}
1.006 &  -0.99 &  0.6
\end{pmatrix}$ and $H=\begin{pmatrix}
    0.9&-0.35&-0.225
\end{pmatrix}\T$. To show the scalability of our protocols, similar to Example 1, we consider three heterogeneous MAS with different number of agents and different communication topologies.

\textit{Case I}: Consider a MAS with $3$ agents with agent models $(C_i, A_i, B_i)$ for $i \in \{1,\hdots,3\}$, and directed communication topology shown in Figure \ref{Tree-3nodes}. Values of communication delays are same as Example 1, case 1.

\textit{Case II}: In this case, we consider a MAS with $5$ agents and agent models $(C_i, A_i, B_i)$ for $i \in \{1,\hdots,5\}$ and directed communication topology shown in Figure \ref{Tree-5nodes}. Values of communication delays are same as Example 1, case 2.

\textit{Case III}: Finally, we consider a MAS with $10$ agents and agent models $(C_i, A_i, B_i)$ for $i \in \{1,\hdots,10\}$ and directed communication topology, shown in Figure \ref{Tree-10nodes}. Values of communication delays are same as Example 1, case 3.\\

The simulation results are shown in Figure \ref{Het-3-dis}-\ref{Het-10-dis}. We observe that our one-shot protocol design works for any MAS with any communication networks $\mathcal{G}\in \mathbb{G}^N$ and any number of agents $N$.

\bibliographystyle{plain}
\bibliography{referenc}

\end{document}